\documentclass[letterpaper,english,preprintnumbers,amsmath,amssymb,superscriptaddress,nofootinbib,prl]{revtex4}
\usepackage{mathpazo}
\usepackage[T1]{fontenc}
\usepackage[latin9]{inputenc}
\usepackage{verbatim}
\usepackage{bbding}
\usepackage{amsmath}
\usepackage{amsthm}
\usepackage{amssymb}
\usepackage{graphicx}

\makeatletter

\pdfpageheight\paperheight
\pdfpagewidth\paperwidth

\providecommand{\tabularnewline}{\\}

\@ifundefined{textcolor}{}
{%
 \definecolor{BLACK}{gray}{0}
 \definecolor{WHITE}{gray}{1}
 \definecolor{RED}{rgb}{1,0,0}
 \definecolor{GREEN}{rgb}{0,1,0}
 \definecolor{BLUE}{rgb}{0,0,1}
 \definecolor{CYAN}{cmyk}{1,0,0,0}
 \definecolor{MAGENTA}{cmyk}{0,1,0,0}
 \definecolor{YELLOW}{cmyk}{0,0,1,0}
}
  \theoremstyle{remark}
  \ifx\thechapter\undefined
    \newtheorem{rem}{\protect\remarkname}
  \else
    \newtheorem{rem}{\protect\remarkname}[chapter]
  \fi
  \theoremstyle{definition}
  \ifx\thechapter\undefined
    \newtheorem{defn}{\protect\definitionname}
  \else
    \newtheorem{defn}{\protect\definitionname}[chapter]
  \fi
  \theoremstyle{plain}
  \newtheorem{assumption}{\protect\assumptionname}
  \theoremstyle{plain}
  \newtheorem*{prop*}{\protect\propositionname}
  \theoremstyle{remark}
  \newtheorem{notation}{\protect\notationname}
 \theoremstyle{definition}
 \newtheorem*{defn*}{\protect\definitionname}
  \theoremstyle{plain}
  \newtheorem*{thm*}{\protect\theoremname}
  \theoremstyle{plain}
  \ifx\thechapter\undefined
    \newtheorem{lem}{\protect\lemmaname}
  \else
    \newtheorem{lem}{\protect\lemmaname}[chapter]
  \fi
  \theoremstyle{plain}
  \ifx\thechapter\undefined
    \newtheorem{thm}{\protect\theoremname}
  \else
    \newtheorem{thm}{\protect\theoremname}[chapter]
  \fi
  \theoremstyle{plain}
  \ifx\thechapter\undefined
    \newtheorem{cor}{\protect\corollaryname}
  \else
    \newtheorem{cor}{\protect\corollaryname}[chapter]
  \fi

\makeatother

\usepackage{babel}
  \providecommand{\assumptionname}{Assumption}
  \providecommand{\definitionname}{Definition}
  \providecommand{\lemmaname}{Lemma}
  \providecommand{\notationname}{Notation}
  \providecommand{\propositionname}{Proposition}
  \providecommand{\remarkname}{Remark}
  \providecommand{\theoremname}{Theorem}
\providecommand{\corollaryname}{Corollary}
\providecommand{\theoremname}{Theorem}

\begin{document}

\title{Eigenvalue approximation of sums of Hermitian matrices from eigenvector
localization/delocalization}

\author{Ramis Movassagh}
\email{q.eigenman@gmail.com}

\affiliation{Department of Mathematics, IBM T. J . Watson Research Center, Yorktown
Heights, NY 10598}

\author{Alan Edelman}
\email{edelman@math.mit.edu}

\affiliation{Department of Mathematics, Massachusetts Institute of Technology,
Cambridge, MA 02139}

\date{\today}
\begin{abstract}
We propose a technique for calculating and understanding the eigenvalue
distribution of sums of random matrices from the known distribution
of the summands. The exact problem is formidably hard. One extreme
approximation to the true density amounts to \textit{classical} probability,
in which the matrices are assumed to commute; the other extreme is
related to \textit{free} probability, in which the eigenvectors are
assumed to be in generic positions and sufficiently large. In practice,
free probability theory can give a good approximation of the density.

We develop a technique based on eigenvector localization/delocalization
that works very well for important problems of interest where free
probability is not sufficient, but certain uniformity properties apply.
The localization/delocalization property appears in a convex combination
parameter that notably, is independent of any eigenvalue properties
and yields accurate eigenvalue density approximations.

We demonstrate this technique on a number of examples as well as discuss
a more general technique when the uniformity properties fail to apply.
\end{abstract}
\maketitle

\section{Summary of the main results}
This paper proposes an answer to an applied mathematics problem with
a rich pure history: what are the eigenvalues of the sum of two symmetric
matrices? Knutson and Tao remind us \cite{knutson2001honeycombs}
that in 1912 Hermann Weyl asked for all the possible eigenvalues that
can result given the eigenvalues of the summands \cite{weyl1912asymptotische}.
We ask a less precise question that we suspect may also be more useful.
What might the spectrum (as a distribution) look like?

Let us start by the eigenvalue decompositions of two $m\times m$
self-adjoint matrices $M_{1}=Q_{1}^{-1}\Lambda_{1}Q_{1}$ and $M_{2}=Q_{2}^{-1}\Lambda_{2}Q_{2}$
where $\Lambda_{1}$ and $\Lambda_{2}$ are diagonal matrices of eigenvalues
of $M_{1}$ and $M_{2}$, and $Q_{1}$ and $Q_{2}$ are $\beta-$orthogonal
matrices with $\beta=1,2,4$ denoting real orthogonal, unitary and
symplectic respectively. The goal then becomes to compute the eigenvalue
distribution of $M\equiv M_{1}+M_{2}$ from the knowledge of the distributions
of $\Lambda_{1}$ and $\Lambda_{2}$.

Let us change basis and write $M_{1}+M_{2}$ as
\begin{equation}
M\equiv\Lambda_{1}+Q_{s}^{-1}\Lambda_{2}Q_{s},\label{eq:problem-1}
\end{equation}
where $Q_{s}\equiv Q_{2}Q_{1}^{-1}$. 

Let us define the classical and finite free versions of this problem,
respectively, by
\begin{eqnarray}
M_{c} & = & \Lambda_{1}+\Pi^{-1}\Lambda_{2}\Pi\label{eq:problem_classical-1}\\
M_{f} & = & \Lambda_{1}+Q^{-1}\Lambda_{2}Q\label{eq:problem_free-1}
\end{eqnarray}
where $\Pi$ denotes a uniform random permutation matrix and $Q$
is a $\beta-$Haar orthogonal matrix. Note that we only replaced the
exact $Q_{s}$ in Eq. \eqref{eq:problem-1} with the appropriate approximations.
That is $\Lambda_{1}$ and $\Lambda_{2}$ are kept the same in $M_{c}$,
and $M_{f}$. 
\begin{figure}
\begin{centering}
\includegraphics[scale=0.4]{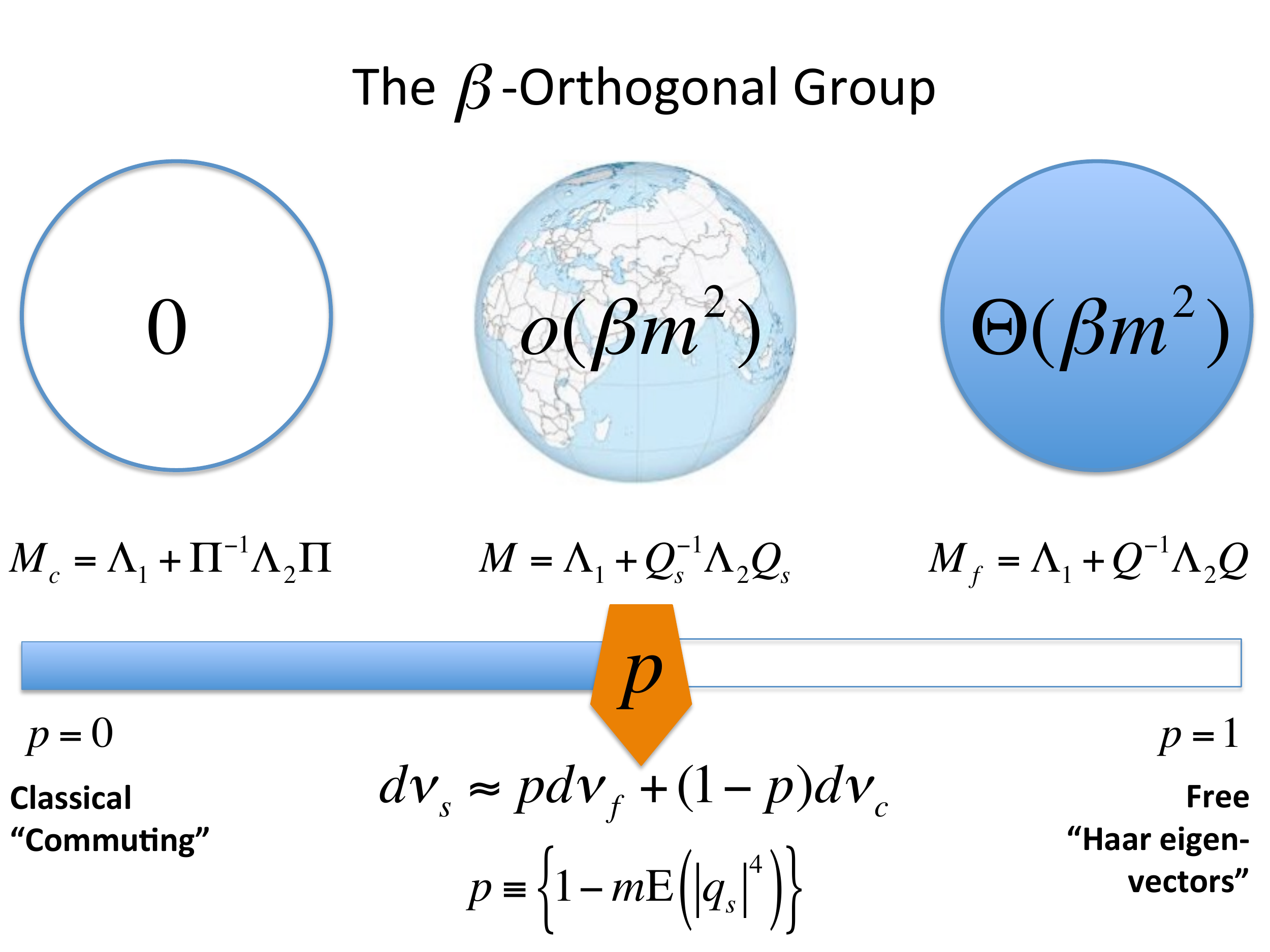}
\par\end{centering}
\caption{\label{fig:Depiction-of-the}Depiction of the proposition. Inside
the spheres we show the parameter count for the corresponding $\beta-$orthogonal
matrix }
\end{figure}

\begin{rem}
The eigenvalue distribution of $M_{c}$ and $M_{f}$ are, respectively,
the classical and finite ``free'' convolution of the distributions
corresponding to $\Lambda_{1}$ and $\Lambda_{2}$. 
\end{rem}
Let $d\nu_{1}$ and $d\nu_{2}$ be the eigenvalue densities of $M_{1}$
and $M_{2}$ respectively. By assumption the distribution of $M$,
denoted by $d\nu_{M}$, is hard to compute. The notation we use for
the classical and finite free convolutions of $d\nu_{1}$ and $d\nu_{2}$
respectively is
\begin{eqnarray}
d\nu^{c} & = & d\nu_{1}\boxplus_{c}\mbox{ }d\nu_{2}\qquad\mbox{classical}\label{eq:dnu1-1}\\
d\nu^{f} & = & d\nu_{1}\boxplus_{f}\mbox{ }d\nu_{2}\qquad\mbox{Free}.\label{eq:dnu2-1}
\end{eqnarray}

Classical approximation assumes that $M_{1}$ and $M_{2}$ commute
(Eq. \eqref{eq:problem_classical-1}), whereas, the free approximation
(Eq. \eqref{eq:problem_free-1}) is the extreme opposite in the sense
that in $M_{f}$ the relative eigenvectors are in completely generic
positions. Moreover, the number of random parameters in $\Pi$ and
$Q$ are the minimum and maximum possible respectively (Fig. \eqref{fig:Depiction-of-the}).
These observations motivate the proposal that the actual problem is
in-between.

There is a line, the convex combination, that connects these two extremes
that is both mathematically natural and in practice very powerful
for obtaining the density of the sum. We denote it by
\begin{equation}
d\nu^{(p)}\equiv d\nu_{1}\boxplus_{p}d\nu_{2}=p\mbox{ }d\nu^{f}+(1-p)\mbox{ }d\nu^{c}\quad\label{eq:ConxCombMeasure-1}
\end{equation}
for $0\le p\le1$. Note that $d\nu^{(0)}\equiv d\nu^{c}$ and $d\nu^{(1)}=d\nu^{f}$.

Many applied problems involve summing random objects whose measures
are $d\nu_{1}$ and $d\nu_{2}$. We hypothesize that very often the
measure of the sum is well approximated by either $d\nu^{(1)}$ or
$d\nu^{(p)}$ for some $0\le p\le1$, where we describe how to obtain
the appropriate parameter $p$. 

We define the $k^{th}$ empirical moment of $M$ by
\begin{equation}
m_{k}=\varphi[M^{k}]=\frac{1}{m}\mathbb{E}\text{Tr}[M^{k}].\label{eq:k_moment-2}
\end{equation}

We find that $\mathbb{E}\text{Tr}(M^{k})=\mathbb{E}\text{Tr}(M_{c}^{k})=\mathbb{E}\text{Tr}(M_{f}^{k})$
for $k=1,2,3$, i.e., the fourth moment is where the three problems
distinguish themselves. Therefore, we define $p$ by matching fourth
moments
\begin{equation}
m_{4}=pm_{4}^{f}+(1-p)m_{4}^{c}\text{ },\label{eq:Match4Moments-1}
\end{equation}
where $m_{4}\equiv\varphi(M^{4})$, $m_{4}^{c}\equiv\varphi(M_{c}^{4})$,
and $m_{4}^{f}\equiv(M_{f}^{4})$ need to be calculated exactly to
solve for $p$, which using the above equation is simply
\begin{equation}
p=\frac{m_{4}^{c}-m_{4}}{m_{4}^{c}-m_{4}^{f}}=\frac{\varphi(M_{c}^{4})-\varphi(M^{4})}{\varphi(M_{c}^{4})-\varphi(M_{f}^{4})}.\label{eq:pExplicit-1}
\end{equation}

So far in this section, the problem setup has been completely general.
An interesting and a surprisingly simple and general formula for $p$
can be derived if we make an assumption (Assumption \eqref{assu:permIndep}).
In practice the domain of applicability of this technique (Eqs. \eqref{eq:ConxCombMeasure-1}
and \eqref{eq:pExplicit-1}) extends beyond.
\begin{defn}
We say the eigenvector matrix $U$ is permutation invariant, when
given two permutation matrices $\Pi_{1}$ and $\Pi_{2}$, the joint
distribution of the entries of $U$ and the joint distribution of
the entries of $\Pi_{1}U\Pi_{2}$ are the same.
\end{defn}
\begin{assumption}
\label{assu:permIndep}In Eq. \eqref{eq:problem-1}, $\Lambda_{1}$
and $\Lambda_{2}$ are independent random diagonal matrices. $Q_{s}$
is random and permutation invariant (but not necessarily Haar).
\end{assumption}
\begin{prop*}
Under this assumption, the eigenvalues density of $M$ is approximated
by $d\nu_{M}\approx d\nu^{(p)}$, where $d\nu^{(p)}\equiv pd\nu^{f}+(1-p)d\nu^{c}$.
The parameter $0\le p\le1$ is defined by
\begin{equation}
p=\frac{m_{4}^{c}-m_{4}}{m_{4}^{c}-m_{4}^{f}}=\frac{\left\{ 1-m\mathbb{E}\left(|q_{s}|^{4}\right)\right\} }{\left\{ 1-m\mathbb{E}\left(|q|^{4}\right)\right\} }\text{ }\overset{m\rightarrow\infty}{=}\text{ }1-m\mathbb{E}(|q_{s}|^{4}),\label{eq:p-1}
\end{equation}
where $q_{s}$ denotes any entry of $Q_{s}$, and $q$ denotes any
entry of the $\beta-$Haar $Q$.
\end{prop*}
We were surprised to find that $p$ is independent of the eigenvalue
distributions and in that sense is universally given by Eq. \eqref{eq:p-1}
as long as the eigenvectors are permutational invariant.
\begin{rem}
In the finite case, in Eq. \eqref{eq:p-1} we have a ratio of $1-m\mathbb{E}\left(|q_{s}|^{4}\right)$
and $1-m\mathbb{E}\left(|q|^{4}\right)$. These are measures of the
localization of the eigenvectors of $Q_{s}$ and $Q$ respectively,
and in physics literature are called inverse participation ratios.
Let us illustrate this by taking a general eigenvector matrix $U$
and denote any column of it by $u$. Denote its entries by $u_{i}$.
Since $\mathbb{E}(|u_{i}|{}^{2})=(1/m)\sum_{i}|u_{i}|^{2}=1$ and,
because of centrality $\mathbb{E}(u_{i}^{3})=0$ , a good measure
for distribution of entries of $u$ is
\[
1-m\mathbb{E}(|u_{i}|^{4})=1-\sum_{i=1}^{m}|u_{i}|^{4}=\left\{ \begin{array}{cccc}
0 & \quad & u=(0,\dots,0,1,0,\dots,0)^{T} & \text{most localized},\\
1-1/m & \quad & u=\frac{1}{\sqrt{m}}(1,1,\dots,1) & \text{most delocalized}.
\end{array}\right.
\]
As $m\rightarrow\infty$ the inverse participation ratio goes to $1$
for the most delocalized eigenvectors. It is fascinating that in quantifying
localization and teasing apart the difference among empirical measures,
the fourth moment is what matters most. 
\end{rem}
\subsection*{Illustration}
We provide two illustrations of this theory that are relevant in quantum
many-body systems (see Fig. \eqref{fig:Illustration000}) and defer
the details and further examples to Section \ref{sec:Illustrations-and-Applications}.
The Figure on the left shows the density of states (DOS) of a quantum
spin chain with generic local interactions in which $p=0.43$. The
example on the right is the DOS of the Anderson model in which $p=1$
(i.e., the free approximation suffices). 
\begin{figure}
\begin{raggedright}
\includegraphics[scale=0.3]{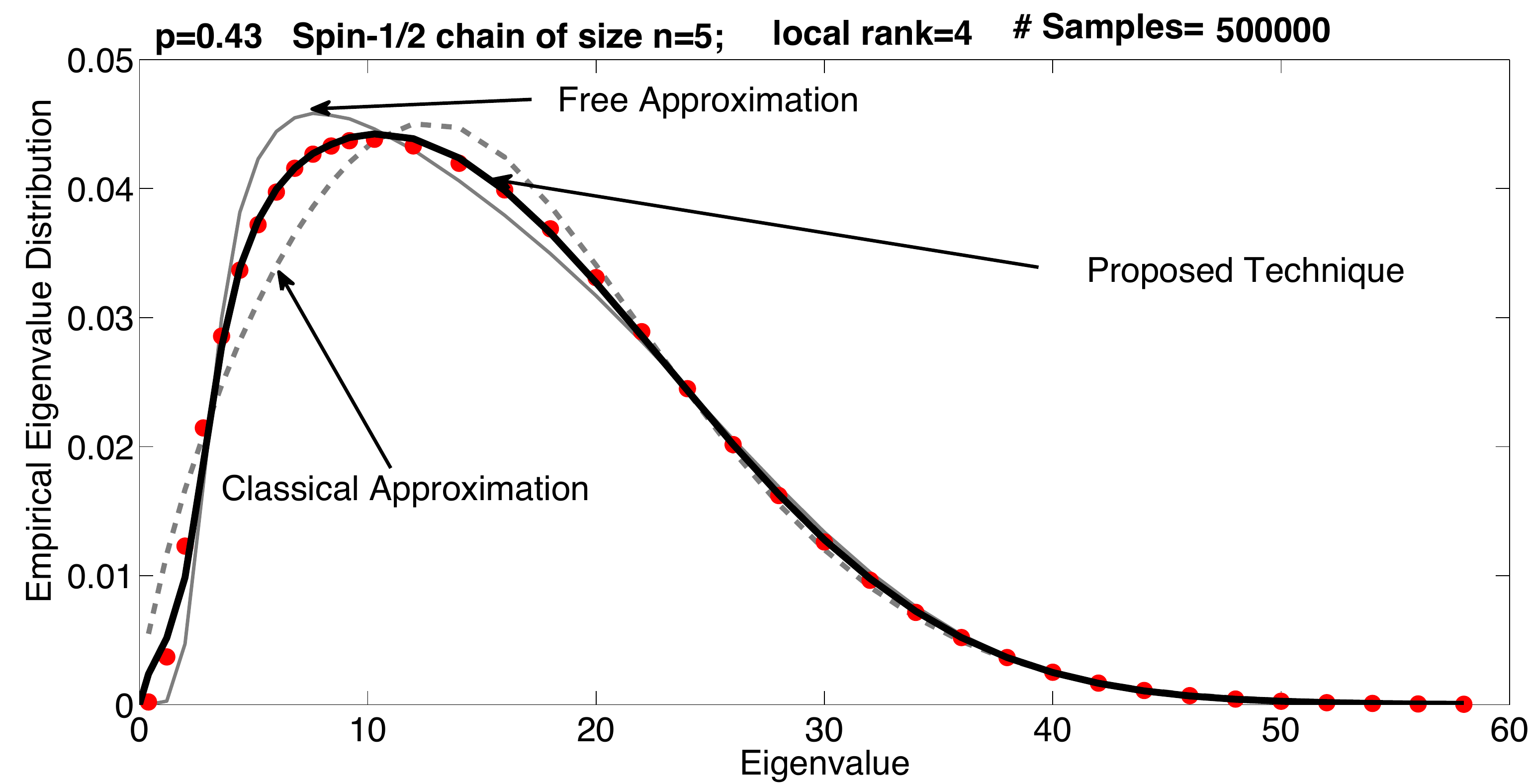}$\qquad$\includegraphics[scale=0.25]{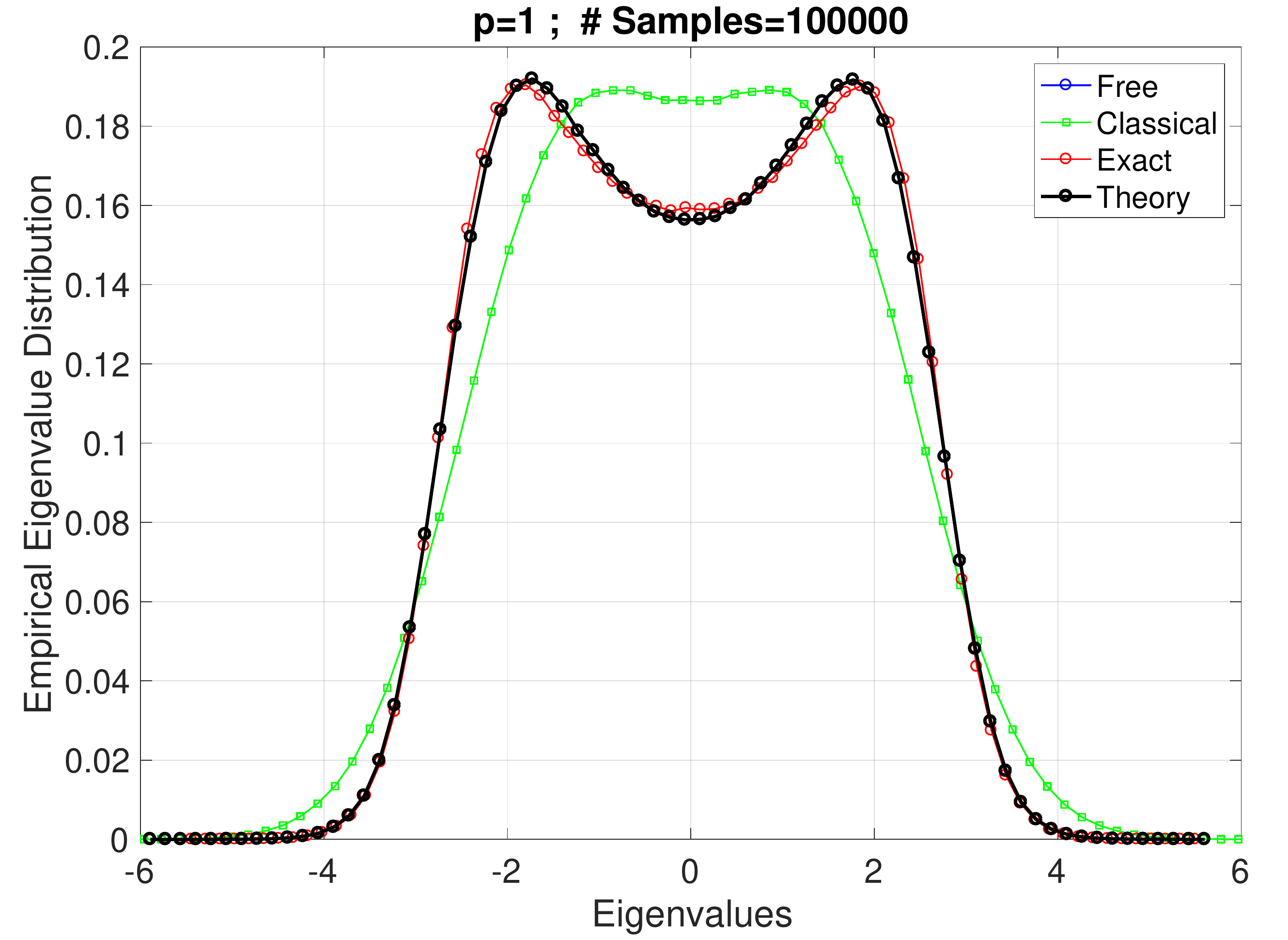}
\par\end{raggedright}
\caption{\label{fig:Illustration000}Left: Density of states (DOS) of a quantum
spin chain with generic local interactions \cite{movassagh2011density}.
Here $p=0.43$, the solid grey curve is the DOS of the free approximation
(i.e, $p=1$) and dashed grey curve is the DOS of classical version
($p=0$). The red dots are the exact DOS and the black solid line
is the approximation obtained from our technique. Right: DOS of the
Anderson model with $p=1$. Sum of a diagonal standard random gaussian
matrix, $M_{1}$, and the hopping matrix, $M_{2}=Q^{T}(2\mathbb{I}+L)Q$,
where $L$ is the Laplacian matrix. Please see Section \ref{sec:Illustrations-and-Applications}
for the details.}
\end{figure}
\section{Introduction}
Given the eigenvalues of two $m\times m$ Hermitian matrices, how
does one determine all the possible set of the eigenvalues of the
sum? As stated at the very beginning of this paper, H. Weyl's question
lead to many mathematical developments and A. Horn's seminal work
that conjectured a (over-complete) set of recursive inequalities for
the eigenvalues of sums of Hermitian matrices \cite{horn1962eigenvalues}.
This conjecture was proved by Klyachko \cite{klyachko1998stable}
and later made clearer with the use of  Schubert calculus by Knudson
and Tao \cite{knutson2001honeycombs}. However, the bounds obtained
from these works are not very good for sparse matrices which are often
encountered in practice (e.g., local Hamiltonians that physicists
often consider). In any case and despite these great successes, there
are not many results that with a high accuracy compute the eigenvalues
of the sum from the knowledge of the summands. 

Our goal is pragmatic: we seek a method that enables us to draw (on
a computer) an accurate picture of the density of the eigenvalues
of the sum from those of the summands.\\

Given the probability measures $d\nu_{1}$ and $d\nu_{2}$ of two
random variables, one can ask: what is the measure of the sum of the
random variables? In classical probability theory in which the random
variables commute, the measure of the sum is the convolution of the
measures. In the other extreme, where the random variables do not
commute and are generic (e.g., random matrices), the measure in the
infinite limit is the free convolution \cite{nica2006lectures,voiculescu1992free}.

Let us define the $\varphi$ notation following \cite{nica2006lectures}.
Let ${\cal A}$ be unital algebras over $\mathbb{C}$ \footnote{Everything goes through the same if the algebra is over reals $\mathbb{R}$
or quaternions $\mathbb{Q}$}. The elements of ${\cal A}$ are in general non-commuting. 
\begin{defn}
\label{def:Phi}Let $\varphi$ be a unital linear functional $\varphi:\text{ }{\cal A}\rightarrow\mathbb{C}$,
with the properties that $\varphi$ is a trace and $\varphi[1_{{\cal A}}]=1$.
$\varphi$ is a trace in the sense that
\[
\varphi[ab]=\varphi[ba],\quad\forall a,b\in{\cal A},
\]
Let $\varphi_{c}$ be the ``commutative'' version of $\varphi$,
that has the additional property that the order of the product of
its arguments do not matter, i.e., $\varphi_{c}[abc]=\varphi_{c}[bac]=\cdots$. 
\end{defn}
\begin{notation}
When the variables (elements of ${\cal A})$ are $m\times m$ matrices,
then
\[
\begin{array}{ccc}
\varphi[\centerdot]\equiv\frac{1}{m}\text{Tr}[\centerdot] &  & \text{non-random matrices}\\
\\
\varphi[\centerdot]\equiv\frac{1}{m}\mathbb{E}\text{Tr}[\centerdot] &  & \text{random matrices .}\\
\\
\varphi[\centerdot]=\int\centerdot\text{ }d\mu &  & \text{operators}
\end{array}
\]
Given a random matrix $M$, the expected empirical measure of its
eigenvalues is
\[
d\nu_{M}=\varphi[M]=\frac{1}{m}\mathbb{E}\left\{ \sum_{i=1}^{m}\delta\left(\lambda-\lambda_{i}(M)\right)\right\} .
\]
\end{notation}
\subsection{Introduction to Free Probability Theory}
Free probability theory (FPT) is suited for \textit{non-commuting} random variables. The more conventional probability theory (CPT)
deals with commuting random variables. 

Supposed $M_{1},M_{2},\cdots,M_{N}$ are $m\times m$ random matrices
with known eigenvalue distributions, what is the eigenvalue distribution
of
\begin{equation}
M=M_{1}+M_{2}+\cdots+M_{N}\quad?\label{eq:prob}
\end{equation}
FPT answers this question if $M_{k}$'s are free. We define free independence
following Nica and Speicher \cite{nica2006lectures}.
\begin{defn}
({[}Nica Speicher{]} Free Independence) Let $({\cal A},\varphi)$
be a non-commutative probability space and let $I$ be a fixed index
set. The subalgebras $({\cal A}_{i})_{i\in I}$ are called \textit{free
independent} with respect to the functional $\varphi$, if 
\[
\varphi(a_{1}\dots a_{k})=0
\]
whenever we have the following:
\end{defn}
\begin{itemize}
\item $k$ is a positive integer;
\item $a_{j}\in{\cal A}_{i(j)}$ , $(i(j)\in I)$ for all $j=1,\dots,k$;
\item $\varphi(a_{j})=0$ for all $j=1,\dots,k$;
\item and neighboring elements are from different subalgebras, i.e., $i(1)\ne i(2)$,
$i(2)\ne i(3)$, $\dots$, $i(k-1)\ne i(k)$. 
\end{itemize}
Recall that in CPT the distribution of sum of random variables is
not additive but the cumulants or log-characteristics are. The analogous
additive quantities in FPT are \textit{free cumulants }and\textit{
$r-$transforms} \cite{nica2006lectures}. 

How can we make utilize FPT to analytically obtain the eigenvalue
distribution of Eq. \eqref{eq:prob}? As long as $M_{k}$'s are free
from one another, theoretically, the free convolution will provide
the distribution of the sum. However, its numerical computation may
be difficult. 

For the sake of concreteness, suppose we have two matrices $M_{1}$
and $M_{2}$ , which may not be free, and we are interested in the
spectrum of the sum
\begin{equation}
M=M_{1}+M_{2};\label{eq:toy}
\end{equation}
the free approximation can be obtained by (possibly slightly) changing
the problem. Mathematically, FPT would obtain the eigenvalue distribution
of
\[
M_{1}+Q^{-1}M_{2}Q
\]
where, $Q$ is an $m\times m$ Haar distributed $\beta-$orthogonal
matrix as before. This amounts to spinning the eigenvectors to point
randomly and uniformly on a sphere in orthogonal group $\mathcal{O}(m)$
uniformly. Our technology can treat both finite and infinite matrices.
One need not use the standard fields; arbitrary number fields can
be used by replacing $Q$ in Eq. \eqref{eq:toy} by the corresponding
Haar matrices (see Table \eqref{tab:Fields}). 
\begin{table}
\begin{centering}
\begin{tabular}{|c|c|c|c|c|}
\hline 
Field & Real & Complex & Quaternions & ``Ghosts''\tabularnewline
\hline 
\hline 
$\beta$ & $1$ & $2$ & $3$ & $>3$\tabularnewline
\hline 
Haar matrices & $Q$ & $U$ & $S$ & $\mathcal{G}$\tabularnewline
\hline 
\end{tabular}
\par\end{centering}
\centering{}\caption{\label{tab:Fields}Notation for various fields of numbers}
\end{table}
\begin{rem}
Standard FPT proves that $M_{1}$ and $Q^{-1}M_{2}Q$ are \textit{asymptotically}
free. If we look at the moments of the sum, i.e., $\varphi\left[M_{1}+Q^{-1}M_{2}Q\right]^{k}=\frac{1}{m}\mathbb{E}\text{Tr}\left[M_{1}+Q^{-1}M_{2}Q\right]^{k}$
then $O(1)$ terms would match the answer that FPT would provide and
there will be additional terms (finite corrections) that will be at
most $O(1/m)$.
\end{rem}
Since its eigenvectors are Haar, one naturally thinks of the free
approximation as the most delocalized. For finite Haar distributed
$\beta-$orthogonal matrices (compare with Eq. \eqref{eq:p-1}), 
\begin{equation}
1-m\mathbb{E}(|q|^{4})=\frac{(m-1)\beta}{m\beta+2}\label{eq:IPR_free}
\end{equation}
which in the limit of $m\rightarrow\infty$ is independent of $\beta$
and equal to one. More generally, for $\beta-$Haar orthogonal matrix
of size $m\times m$ we have
\begin{center}
\begin{tabular}{|l|c|}
\hline 
\multicolumn{2}{|c|}{Moments of $\beta-$Haar Orthogonal matrix }\tabularnewline
\hline 
\hline 
Expected values & Count\tabularnewline
\hline 
\hline 
$\mathbb{E}(|q_{i,j}|^{2})=1/m$ & $m^{2}$\tabularnewline
\hline 
$\mathbb{E}(|q_{i,j}|^{4})=\frac{\beta+2}{m(m\beta+2)}$ & $m^{2}$\tabularnewline
\hline 
$\mathbb{E}(|q_{i,j}|^{2}|q_{i,k}|^{2})=\frac{\beta}{m(m\beta+2)}$,
$j\ne k$ & $2m^{2}(m-1)$\tabularnewline
\hline 
$\mathbb{E}(\bar{q}_{ji}q_{jk}\bar{q}_{pk}q_{pi})=\frac{-\beta}{m(m\beta+2)(m-1)}$,
$i\ne k\text{ and }j\ne p$ & $m^{2}(m-1)^{2}$\tabularnewline
\hline 
\end{tabular}
\par\end{center}

Comment: These formulas can be derived from Weingarten formulas or
direct calculations for $\beta=1,2,4$. We have checked the quantities
in the table above against numerical experiments for $\beta=1,2$.
General $\beta\notin\{1,2,4\}$ is a subject of current speculation.
\section{More than two matrices}
In our work we satisfy ourselves with sums of two hermitian matrices.
However, in the next two subsections we provide results that extend
the moment computation for the classical and free modifications of
the problem.
\subsection{Classical irreducible moment expansion}
\begin{defn}
(Classically Equivalent) In the expansion of $\varphi_{c}\left[\left(M_{1}+M_{2}+\dots+M_{k}\right)^{n}\right]$,
there are $k^{n}$ monomials that can be put into distinct equivalent
classes under $\varphi_{c}$. Each equivalence class is defined by
the distinct set of positive integers $j_{i}\in[n]$ for $1\le i\le k$,
where any fixed set $j_{1},\dots,j_{k}$ corresponds to the number
of times $M_{1},\dots,M_{k}$ appear in the expansion respectively.

Because of the commutativity, the binomial theorem can be evoked,
and by the cyclic property of $\varphi_{c}$
\begin{equation}
\varphi_{c}\left[\left(M_{1}+M_{2}\right)^{n}\right]=\sum_{j=0}^{n}\left(\begin{array}{c}
n\\
j
\end{array}\right)\varphi_{c}\left[M_{1}^{j}M_{2}^{n-j}\right],\label{eq:ClassicalEquiv}
\end{equation}
where each summand is the contribution of the $j^{\text{th}}$ equivalent
class. More generally, 
\[
\varphi_{c}\left[\left(M_{1}+M_{2}+\cdots+M_{k}\right)^{n}\right]=\sum_{j_{1}+\cdots+j_{k}=n}\left(\begin{array}{ccccc}
 &  & n\\
j_{1} & , & \cdots & , & j_{k}
\end{array}\right)\varphi_{c}\left[M_{1}^{j_{1}}M_{2}^{j_{2}}\dots M_{k}^{j_{k}}\right],
\]
where each summand once again is the contribution of one of the equivalent
classes. 
\end{defn}
We wish to generalize these classical notions to the \textit{non-commutative}
setting, whereby the reduced form of the non-classical (i.e., non-commuting)
moment expansion $\varphi\left[\left(A+B\right)^{n}\right]$ is found.
As a first step, it would be helpful to know the number of terms of
each type that are cyclically equivalent with respect to $\varphi$.
\subsection{Free irreducible moment expansion}
\begin{defn}
(trace-equivalent) In the general non-commuting $n^{th}$ moment expansion
\begin{equation}
\varphi\left[(M_{1}+M_{2}+\cdots+M_{k})^{n}\right]\label{eq:k_moment-1}
\end{equation}
there are $n^{k}$ monomials each of which is a product of $n$ terms
chosen from the alphabet $\{M_{1},M_{2},\dots,M_{k}\}$. We define
each trace-equivalent class to be the subset of monomials that are
equal under $\varphi$.
\end{defn}
So how many of such equivalent classes are there? The answer to this
question is equivalent to a theorem by Polya \cite{riordan2012introduction}. 
\begin{defn*}
An $(n,k)$-necklace is an equivalence class of words of length $n$
over an alphabet of size $k$ under rotation (i.e., cyclically equivalent).
The total number of such distinct necklaces is denoted by $a(n,k)$. 
\end{defn*}
\begin{thm*}
(Polya) Let $\phi(d)$ be the Euler function of the positive integer
$d$ and $d|n$ denote all the divisors of the integer $n$ then
\[
a(n,k)=\frac{1}{n}\sum_{i=1}^{n}k^{\text{gcd}(n,i)}=\frac{1}{n}\sum_{d|n}\phi(d)k^{n/d}.
\]
\end{thm*}
For example, in $\varphi\left[(M_{1}+M_{2})^{n}\right]$ there are
$a(n,2)$ necklaces. More generally, in $\varphi\left[(M_{1}+M_{2}+\cdots+M_{k})^{n}\right]$
there are $a(n,k)$ necklaces. In Fig. \eqref{Fig:PolyaIllust} we illustrate
the equivalent classes of $a(3,2)$ and $a(4,2)$. The former corresponds
to $\varphi\left[(M_{1}+M_{2})^{3}\right]$ and the latter to $\varphi\left[(M_{1}+M_{2})^{4}\right]$.
\begin{figure}
\centering{}\includegraphics[scale=0.4]{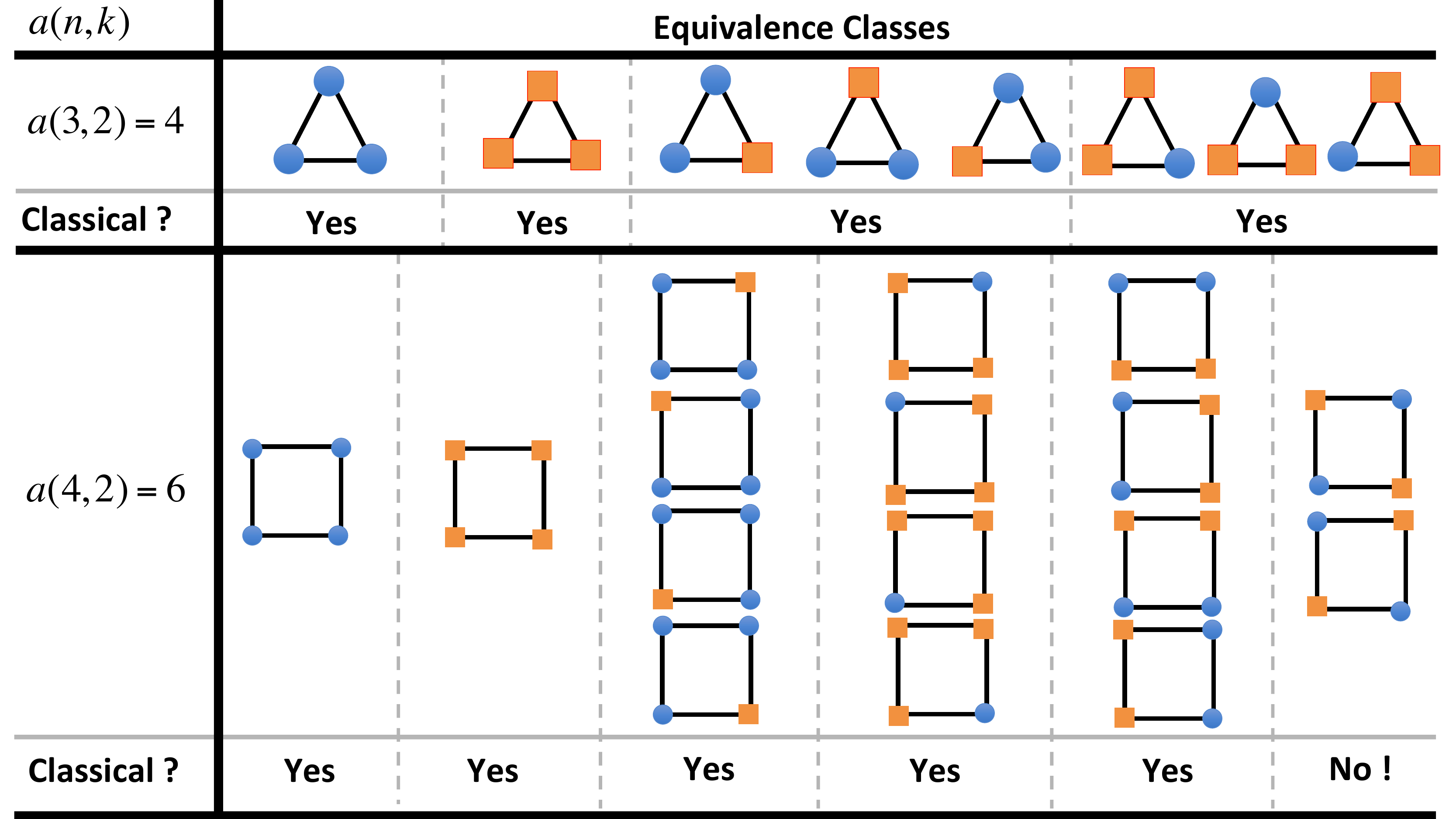}\caption{\label{Fig:PolyaIllust}Illustration of Theorem. The one-to-one correspondence
with moment expansion in Eq. \eqref{eq:k_moment-1} can be done by thinking
of circles as $M_{1}$ and squares as $M_{2}$. }
\end{figure}
\begin{lem}
In the expansion $\varphi\left[(M_{1}+M_{2})^{n}\right]$ there are
only $(n-1)^{2}+1$ terms that are classical. 
\end{lem}
\begin{proof}
These would coincide with the terms that are cyclically equal to $\varphi_{c}\left[M_{1}^{j}M_{2}^{n-j}\right]$.
Suppose $M_{1}$ appears $j$ times. If $0<j<n$ the length of the
cyclic orbit is exactly $n$. However, if $j=0$ or $j=n$, then there
is no orbit and each has exactly one term in the corresponding equivalence
class. We have altogether $(n-2)n+2$ classical terms. 
\end{proof}
\section{Technical Results}
We now return to the problem of approximating the eigenvalue distribution
of sums of two hermitian matrices. Below we use $U$ to denote an
eigenvector matrix that is permutation invariant and $\beta-$orthogonal;
it can be $Q_{s}$, $\Pi$ or $Q$. That is we reserve $U$ when the
results being proved do not depend on the choice of the three cases.
We assume that the columns of $U$ are chosen so that each column
and its negation are equiprobable. One consequence is that the mean
of every element of $U$ is zero. Below repeated indices are summed
over unless states otherwise. We denote the (diagonal) entries of
$\Lambda_{1}$ and $\Lambda_{2}$ by
\begin{eqnarray*}
\lambda_{i} & \rightarrow & \Lambda_{1}\\
\mu_{i} & \rightarrow & \Lambda_{2}
\end{eqnarray*}
\begin{lem}
\label{lem:var}The elements of $U$ are (dependent) random variables
with mean zero and variance $1/m$.
\end{lem}
\begin{proof}
The invariance under the change of sign implies the zero mean. The
variance is $\frac{1}{m^{2}}\sum|u_{ij}|^{2}=1/m$.
\end{proof}
\begin{lem}
\label{lem:(departure-lemma)-.}(departure lemma) $\varphi\left[\Lambda_{1}^{k_{1}}U^{-1}\Lambda_{2}^{k_{2}}U\right]=\mathbb{E}\left(\lambda_{i}^{k_{1}}\right)\mathbb{E}\left(\mu_{j}^{k_{2}}\right)$. 
\end{lem}
\begin{proof}
By permutation invariance and Lemma \eqref{lem:var}, $\mathbb{E}(|u_{ij}|^{2})=1/m$
. For integers $k_{1}>0$ and $k_{2}>0$, we have $\mathbb{E}\textrm{Tr}\left(\Lambda_{1}^{k_{1}}U^{-1}\Lambda_{2}^{k_{2}}U\right)=\mathbb{E}\left(\sum_{i,j}|u_{ij}|^{2}\lambda_{i}^{k_{1}}\mu_{j}^{k_{2}}\right)$.
By the independence of the eigenvalues from the eigenvectors, this
expected value is equal to $m^{2}\left(\frac{1}{m}\right)\mathbb{E}\left(\lambda_{i}^{k_{1}}\mu_{j}^{k_{2}}\right)$
for any $i$ or $j$. By Def. \eqref{def:Phi}, we now have
\[
\varphi\left[\Lambda_{1}^{k_{1}}U^{-1}\Lambda_{2}^{k_{1}}U\right]=\mathbb{E}\left(\lambda_{i}^{k_{1}}\mu_{j}^{k_{2}}\right)=\mathbb{E}\left(\lambda_{i}^{k_{1}}\right)\mathbb{E}\left(\mu_{j}^{k_{2}}\right),
\]
where the last equality follows from the independence of $\Lambda_{1}$
and $\Lambda_{2}$. 
\end{proof}
\begin{lem}
\label{lem:ThreeMomentMatching}The first three moments of $\Lambda_{1}+U^{-1}\Lambda_{2}U$
are equal (and independent of the distribution of $U$).
\end{lem}
\begin{proof}
Using the trace property $\text{Tr}(AB)=\text{Tr}(BA)$, the first
three moments are
\[
\begin{array}{c}
m_{1}\equiv\varphi\left[\Lambda_{1}+\Lambda_{2}\right]\qquad\text{in all three cases}\\
m_{2}\equiv\varphi\left[\left(\Lambda_{1}+U^{-1}\Lambda_{2}U\right)^{2}\right]=\varphi\left(\Lambda_{1}^{2}+2\Lambda_{1}U^{-1}\Lambda_{2}U+\Lambda_{2}^{2}\right)\\
m_{3}\equiv\varphi\left[\left(\Lambda_{1}+U^{-1}\Lambda_{2}U\right)^{3}\right]=\varphi\left(\Lambda_{1}^{3}+3\Lambda_{1}^{2}U^{-1}\Lambda_{2}U+3\Lambda_{1}U^{-1}\Lambda_{2}^{2}U+\Lambda_{2}^{3}\right),
\end{array}
\]
By linearity of the $\mathbb{E}\text{Tr}(\centerdot)$ and Lemma \eqref{lem:(departure-lemma)-.}
$m_{1}$, $m_{2}$ and $m_{3}$ above are all equal to the corresponding
classical first, second and third moments respectively. 
\end{proof}
The fourth moments of the three cases will differ because of the appearance
of the terms that we put in bold-faced and underlined in. 
\begin{figure}
\includegraphics[scale=0.55]{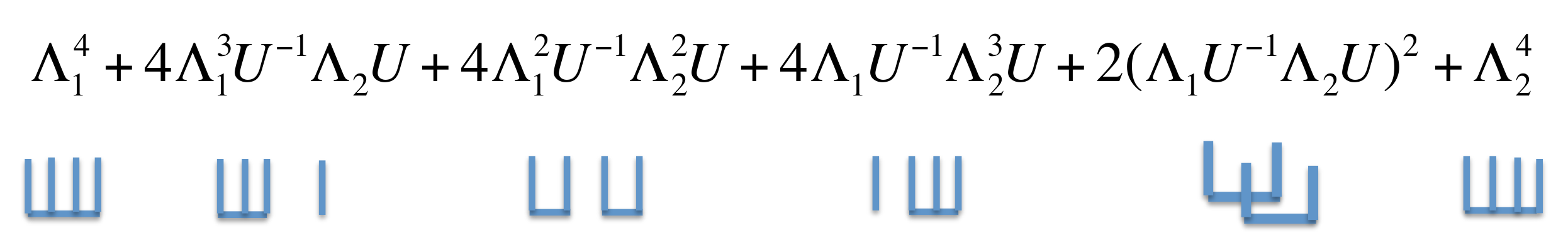}
\caption{\label{fig:Crossing}Fourth moment expansion and (non)-crossing partitions.
We represent each $\Lambda_{1}$ with a vertical line and connect
products of $\Lambda_{1}$ by a horizontal line (similarly with $U^{-1}\Lambda_{2}U$).
Note that $(\Lambda_{1}U^{-1}\Lambda_{2}U)^{2}$ is the only \textit{crossing}
partition, and is the only term whose expected trace differs among
the three cases. The crossing partitions appear first in the fourth
moment expansions. }
\end{figure}
\begin{equation}
m_{4}\equiv\varphi\left[\left(\Lambda_{1}+U^{-1}\Lambda_{2}U\right)^{4}\right]=\varphi\left[\Lambda_{1}^{4}+4\Lambda_{1}^{3}U^{-1}\Lambda_{2}U+4\Lambda_{1}^{2}U^{-1}\Lambda_{2}^{2}U+4\Lambda_{1}U^{-1}\Lambda_{2}^{3}U+\mathbf{\underline{2\left(\mathbf{\Lambda_{1}U^{-1}\Lambda_{2}U}\right)^{2}}}+\Lambda_{2}^{4}\right].\label{eq:fourthMoment}
\end{equation}
In Fig. \eqref{fig:Crossing} we express these terms in their natural
combinatorial representation in terms of (non)-crossing partitions.

Let the symmetric polynomials of degree $k$ in $m$ variables be
denoted by $SP(k,m)$. Moreover let $\lor$ denote a symmetric product,
which we take to mean that the product is invariant under exchange,
i..e, $x\lor y=y\lor x$. Moreover, let $\kappa_{2}(\Lambda)=\frac{1}{m}\sum_{i}\lambda_{i}^{2}-\left(\frac{\sum\lambda_{i}}{m}\right)^{2}$,
which is the ``variance''\footnote{We denote by $m_{2}^{\lambda}=\mathbb{E}(\lambda^{2})$, $m_{1,1}^{\lambda}=\mathbb{E}(\lambda_{i}\lambda_{j})$
and similarly for $m_{2}^{\mu}$ and $m_{1,1}$. They are computed
by
\[
m_{2}^{\lambda}=\frac{1}{m}\sum_{i=1}^{m}\lambda_{i}^{2},\quad m_{1,1}^{\lambda}=\frac{1}{m(m-1)}\sum_{i\ne j}\lambda_{i}\lambda_{j}.
\]
} of $(\lambda_{1},\lambda_{2},\dots,\lambda_{m})$. 
\begin{lem}
\label{lem:SymmetricPoly}Let $f(\Lambda_{1},\Lambda_{2})\in SP(k,m)\lor SP(k,m)$
such that $f(\mathbb{I},\Lambda_{2})=0$. Then if $k=2$, $f(\Lambda_{1},\Lambda_{2})=c\text{ }\kappa_{2}(\Lambda_{1})\kappa_{2}(\Lambda_{2})$,
where $c$ is a constant. 
\end{lem}
\begin{proof}
It is clear that $f$ as a polynomial in $\Lambda_{1}$ is a multiple
of $\kappa_{2}(\Lambda)$ because $f$ vanishes at $\Lambda_{1}=\mathbb{I}$
and $SP(2,m)$ is only two-dimensional. Similarly $f$ is a multiple
of $\kappa_{2}(\Lambda_{2})$ as a polynomial in $\Lambda_{2}$. Since
$f$ vanishes at $\Lambda_{2}=\mathbb{I}$ the only polynomials in
$SP(2,m)\lor SP(2,m)$ with this property are multiples of $\kappa_{2}(\Lambda_{1})\kappa_{2}(\Lambda_{2})$. 
\end{proof}
The following lemma is key:
\begin{lem}
\label{lem:key}$\varphi\left[\left(\Lambda_{1}+\Pi^{-1}\Lambda_{2}\Pi\right)^{4}-\left(\Lambda_{1}+U^{-1}\Lambda_{2}U\right)^{4}\right]=\kappa_{2}(\Lambda_{1})\kappa_{2}(\Lambda_{2})\left\{ 1-m\mathbb{E}\left(|u_{ij}|^{4}\right)\right\} .$ 
\end{lem}
\begin{proof}
$U=\Pi$ is trivial so we think of $U$ as a place holder for $Q_{s}$
and $Q$. Because of the linearity of $\varphi$ and Lemma \eqref{lem:(departure-lemma)-.}
the general form of this difference is
\begin{eqnarray}
\varphi\left[\left(\Lambda_{1}+\Pi^{-1}\Lambda_{2}\Pi\right)^{4}-\left(\Lambda_{1}+U^{-1}\Lambda_{2}U\right)^{4}\right] & = & 2\varphi\left[\left(\Lambda_{1}\Pi^{-1}\Lambda_{2}\Pi\right)^{2}-\left(\Lambda_{1}U^{-1}\Lambda_{2}U\right)^{2}\right],\label{eq:isotropic}
\end{eqnarray}
where the expectation is taken with respect to the random permutations
$\Pi$ and eigenvectors $U$ respectively. 

In Eq. \eqref{eq:isotropic} if $\Lambda_{1}\rightarrow\alpha\Lambda_{1}$
then the right hand side gets multiplied by $\alpha^{2}$, so it is
a homogenous polynomial of second order. Since conjugating either
$\Lambda_{1}$ or $\Lambda_{2}$ by any permutation matrix leaves
the expected trace invariant, the expression is a symmetric polynomial
in entries of $\Lambda_{1}$ and $\Lambda_{2}$. Therefore, by Lemma
\eqref{lem:SymmetricPoly}, we have
\[
\varphi\left[\left(\Lambda_{1}\Pi^{-1}\Lambda_{2}\Pi\right)^{2}-\left(\Lambda_{1}U^{-1}\Lambda_{2}U\right)^{2}\right]=c(U)\text{ }\kappa_{2}(\Lambda_{1})\kappa_{2}(\Lambda_{2}).
\]
To evaluate $c\left(U\right)$, it suffices to let $\Lambda_{1}$
and $\Lambda_{2}$ be projectors of rank one where $\Lambda_{1}$
would have only one nonzero entry on the $i^{\mbox{th }}$ position
on its diagonal and $\Lambda_{2}$ only one nonzero entry on the $j^{\mbox{th }}$
position on its diagonal. Further take those nonzero entries to be
ones, giving $m_{1,1}^{\lambda}=m_{1,1}^{\mu}=0$ and $m_{2}^{\lambda}=m_{2}^{\mu}=1/m$,
and we have
\begin{equation}
\varphi\left[\left(\Lambda_{1}\Pi^{-1}\Lambda_{2}\Pi\right)^{2}-\left(\Lambda_{1}U^{-1}\Lambda_{2}U\right)^{2}\right]=\frac{1}{m^{2}}\text{ }c\left(U\right).\label{eq:fQ}
\end{equation}
But the left hand side is
\begin{eqnarray*}
\varphi\left[\delta_{ij}-|u_{ij}|^{4}\right] & = & \frac{1}{m}\left\{ \frac{1}{m^{2}}\sum_{ij}\delta_{ij}-\frac{1}{m^{2}}\sum_{ij}\mathbb{E}\left(|u_{ij}|^{4}\right)\right\} =\frac{1}{m}\left\{ \frac{1}{m}-\mathbb{E}\left(|u_{ij}|^{4}\right)\right\} ,
\end{eqnarray*}
where we used the homogeneity of $U$. Consequently, by equating this
to $c\left(U\right)/m^{2}$, we get the desired quantity
\[
c\left(U\right)=\left\{ 1-m\mathbb{E}\left(|u_{ij}|^{4}\right)\right\} .
\]
Our final result, i.e., Eq. \eqref{eq:isotropic}, reads
\begin{equation}
\varphi\left[\left(\Lambda_{1}\Pi^{-1}\Lambda_{2}\Pi\right)^{2}-\left(\Lambda_{1}U^{-1}\Lambda_{2}U\right)^{2}\right]=\kappa_{2}(\Lambda_{1})\kappa_{2}(\Lambda_{2})\left\{ 1-m\mathbb{E}\left(|u_{ij}|^{4}\right)\right\} .\label{eq:isotropicBook}
\end{equation}
where $\kappa_{2}(\Lambda_{1})=\left(m_{2}^{\lambda}-m_{1,1}^{\lambda}\right)$
and $\kappa_{2}(\Lambda_{2})=\left(m_{2}^{\mu}-m_{1,1}^{\mu}\right)$
as before.
\end{proof}
\begin{thm}
\label{thm:(universality-of-)}(universality of $p$) In defining
$p$ by  matching fourth moments via $m_{4}=pm_{4}^{f}+(1-p)m_{4}^{c}$,
we find that $p$ is independent of the eigenvalues and is given by
\begin{equation}
p=\frac{m_{4}^{c}-m_{4}}{m_{4}^{c}-m_{4}^{f}}=\frac{\left\{ 1-m\mathbb{E}\left(|q_{s}|^{4}\right)\right\} }{\left\{ 1-m\mathbb{E}\left(|q|^{4}\right)\right\} }\text{ }\overset{m\rightarrow\infty}{=}\text{ }\left\{ 1-m\mathbb{E}\left(|q_{s}|^{4}\right)\right\} ,\label{eq:pTheorem}
\end{equation}
where, as before, $q_{s}$ and $q$ denote any entry of $Q_{s}$ and
$Q$ respectively (see Eqs. \eqref{eq:problem-1} and \eqref{eq:problem_free-1}).
\end{thm}
\begin{proof}
The first equality follows the definition of $p$ via  fourth moment
matching. The second equality follows Lemma \eqref{lem:key} , where
the dependence on eigenvalues as well as an overall factor of $2$
that appear in the numerator and the denominator cancel. The last
equality follows Eq. \eqref{eq:IPR_free} in the limit of $m\rightarrow\infty$,
which corresponds to free probability theory.
\end{proof}
\begin{cor}
\label{thm:(Slider-Theorem)-In}(Slider) $0\le p\le1$. 
\end{cor}
\begin{proof}
Since by normality of eigenvectors $\sum_{i=1}^{m}|q_{s}^{i}|^{2}=1$,
we have that $0\le\sum_{i=1}^{m}|q_{s}^{i}|^{4}\le1$. Now $m\mathbb{E}\left(|q_{s}|^{4}\right)=m\left(\frac{1}{m}\sum_{i=1}^{m}|q_{s}^{i}|^{4}\right)$.
So we have that $0\le1-m\mathbb{E}\left(|q_{s}|^{4}\right)\le1$.
\end{proof}
Comment: $p$ can analytically be calculated if one computes $\mathbb{E}(|q_{s}|^{4})$.
This for example has been done for quantum spin chains with generic
interactions \cite{movassagh2011density,movassagh_IE2010}.%
\begin{rem}
Often in applications, one of the summands is a perturbation of the
other. Namely, $M=M_{1}+\epsilon M_{2}$, where $\left\Vert M_{1}\right\Vert =\left\Vert M_{2}\right\Vert $
and $\epsilon\ll1$. From the analysis above it should be clear that
$p$ is independent of $\epsilon$. 
\end{rem}
\section{Computation of the Density}
The eigenvalue distribution of the  classical extreme is simple; one
simply takes the convolution of the density of the summands. Less
known and more difficult is the computation of the density of the
free sum. Mathematically this is done by taking the free convolution
via the $R-$transform (See \cite{nica2006lectures} for a detailed
discussion). However, the actual computation of the free convolution
is subtle. Olver and Rao made a numerical package that works well
in computing the free convolution under the assumption that the eigenvalue
distribution of the summands has a connected support (it does not
work as well when the support has disjoint intervals) \cite{olver2012numerical}.
Below we provide a complementary method for calculating the free convolution
when the eigenvalues are discrete.
\subsection{Density of the free sum}
Suppose we seek the density of $M$ in Eq. \eqref{eq:prob} under the
assumption that $M_{1},M_{2},\dots,M_{N}$ are free. This, as stated
above, requires the matrices to be infinite in size. In practice,
however, finite (e.g., $30\times30$) random matrices act free. 

One could fix a given matrix $M_{0}$ and take an $N-$fold free sum
of it and ask: What is the density of $M$ when
\begin{equation}
M=Q_{1}^{\dagger}M_{0}Q_{1}+Q_{2}^{\dagger}M_{0}Q_{2}+\dots+Q_{N}^{\dagger}M_{0}Q_{N},\label{eq:FreeM0}
\end{equation}
and each $Q_{i}$ is a $\beta-$Haar orthogonal matrix?

We now define a few important ingredients and outline how the density
of a free sum is computed in theory. The Cauchy transform of any function,
$f(x)$, is given by
\begin{equation}
G(z)=\frac{1}{2\pi i}\int_{\mathbb{R}}dx\text{ }\frac{f(x)}{z-x}\text{ },\label{eq:cachy}
\end{equation}
where for our purposes we use the density $f_{k}(x)$ which denotes
the distribution of the eigenvalues of $M_{k}$ in Eq. \eqref{eq:prob}
(each summand is assumed to be free).

In conventional probability theory, the log-characteristics and cumulants
are additive. In free probability theory, the so called $R-$transform
is additive.

Using the Cauchy transform $G_{k}(z)$, the $R-$transform is defined
by 
\begin{equation}
R_{k}\left(G_{k}(z)\right)=z-\frac{1}{G_{k}(z)},\label{eq:r-transform}
\end{equation}
where in order to obtain $z$, the Cauchy transform Eq. \eqref{eq:cachy}
needs to be inverted. It is good practice to let $w_{k}\equiv G_{k}(z)$,
by which Eq. \eqref{eq:r-transform} reads,
\begin{equation}
R_{k}(w_{k})=G_{k}^{-1}(w_{k})-\frac{1}{w_{k}};\label{eq:r-transform2}
\end{equation}
in solving for $z$ in $w_{k}=G_{k}(z)$, among multiple roots one
chooses the one that is consistent with $\lim_{z\rightarrow\infty}w\sim1/z$. 

Given that we find a way of inverting Eq. \eqref{eq:cachy}, we have
in our hands the $R-$transform of each summand. 

Comment: The inversion  may be tedious. See the next section for a
routine for doing so efficiently.

Let us denote the density of the sum by $f(x)$ and its $R-$transform
by $R(w)$. As stated above, it is a fact of FPT that the $R-$transforms
of the sum are additive \cite{nica2006lectures}. We have
\begin{equation}
R(w)=\sum_{k=1}^{N}R_{k}(w)\overset{\textrm{i.d.}}{=}NR_{0}(w)\label{eq:sumR}
\end{equation}
where the last equality only holds if each $M_{k}$ has identically
distributed eigenvalues, whose $R-$transform is denoted by $R_{0}(w)$.
The last equality also applies in the case of Eq. \eqref{eq:FreeM0}
where each $M_{i}=M_{0}$.

Now we have at our disposal the $R-$transform of the sum and from
it we want to infer the density $f(x)$. The inverse Cauchy transform
of $R(w)$ is
\[
G^{-1}(w)=R(w)+\frac{1}{w}.
\]
The distribution satisfies
\[
w\equiv G(z)=\int_{\mathbb{R}}\frac{f(x)}{G^{-1}(w)-x}\text{ }dx.
\]
Since $G^{-1}(w)$ introduces a branch cut on the real line, we perform
analytical continuation into the complex plane. Let $g^{+}(z)$ be
located right above the branch cut. The distribution is calculated
using Plemelj-Sokhotsky formula:
\begin{equation}
f(x)=\frac{1}{\pi}\lim\left[\textrm{Im}(g^{+}(z)\right].\label{eq:f_x_Plamej}
\end{equation}

This completes the procedure for finding the density of the free sum
of $N$ matrices.
\begin{rem}
The \textit{discrete} Cauchy transform of the spectrum of $M_{k}$
is $G_{k}(z)=\frac{1}{m}\underset{i=1}{\sum^{m}}\frac{1}{z-\lambda_{i}(M_{k})}$,
where $\lambda_{i}(M_{k})$ is an eigenvalue of $M_{k}$. However,
inverting each of the Cauchy transforms involves finding the roots
of a high order complex polynomial, which can be quite difficult.
In subsection \ref{subsec:The-Root-Finding}, we provide a routine
that finds the roots efficiently without solving the high degree polynomial. 
\end{rem}
\subsection{\label{subsec:The-Root-Finding}Detailed Algorithm for discrete spectra}
Suppose we have a discrete distribution
\begin{equation}
f(x)=\frac{1}{m}\sum_{i=1}^{m}\delta(x-\lambda_{i})\label{eq:f_disc}
\end{equation}
and we want the free probability distribution of a random variables
that is distributed according to an $N$-fold sum of random variables,
each of which is distributed according to $f(x)$. More explicitly,
suppose $M_{0}$ is distributed according to $f(x)$ in Eq. \eqref{eq:f_disc}
and we want the distribution of Eq. \eqref{eq:FreeM0} under the assumption
that the eigenvalues of $M_{0}$ are a finite and discrete set. We
now show how to obtain this by using free probability theory as an
approximation.

The Cauchy distribution of $f(x)$ becomes
\begin{equation}
G(z)\equiv w=\frac{1}{m}\sum_{i=1}^{m}\frac{1}{z-\lambda_{i}}.\label{eq:CauchyDisct}
\end{equation}
By the definition of the $R-$transform we can eliminate $z$ by
\[
z=R(w)+\frac{1}{w};
\]
we are interested in an $N-$fold free sum which by additivity of
the $R$-transform amounts to $R(w)\rightarrow R(w)/N$. 

The $z$ of the sum is therefore
\[
z=\frac{R(w)}{N}+\frac{1}{w};
\]
if one were to solve for $R(w)$, one would obtain the $R-$transform
of the sum of $N$ copies the random variables. 

The above procedure can succinctly be performed by only doing the
following transformation on the $z$ of a single random variable
\[
z\rightarrow\frac{z}{N}+\frac{1}{w}\left(1-\frac{1}{N}\right),
\]
where the right hand side is the inverse Cauchy transform of the sum
denoted by $G_{N}^{-1}(w)$.

The discrete inverse Cauchy transform of a sum of $N$ copies of $m\times m$
matrix $M_{0}$ (Eq. \eqref{eq:CauchyDisct}) now reads
\begin{equation}
F(w,z,m)\equiv-w+\frac{1}{m}\sum_{i=1}^{m}\frac{1}{\frac{z}{N}+\frac{1}{w}\left(1-\frac{1}{N}\right)-\lambda_{i}}=0\label{eq:root}
\end{equation}

This is the desired formula. To get the density one applies the Plemelj-Sokhotsky
formula; i.e., one solves for $w$ at a fixed z, and take the imaginary
part and divide by $\pi$.

Solving for $w$ as a function of $z$ requires solving a high degree
polynomial, which may analytically be impossible for polynomials of
degree higher than four. 

After dividing through by $w$, Eq. \eqref{eq:root} can be rewritten
 as
\begin{equation}
\sum_{i=1}^{m}\frac{v_{i}}{w-v_{i}}=\alpha,\label{eq:poles}
\end{equation}
where $v_{i}=\frac{N-1}{N\lambda_{i}-z}$ and denote the poles in
Eq. \eqref{eq:root} and $\alpha=-m\left(\frac{N-1}{N}\right)<0$. 

The solutions of Eq. \eqref{eq:poles} correspond to the intersection
of the horizontal line located at $\alpha$ with $\sum_{i=1}^{m}v_{i}/(w-v_{i})$;
the latter is plotted in Fig. \eqref{fig:Poles}. When all the $v_{i}$
are positive or negative, there are in general exactly $m$ solutions
to Eq. \eqref{eq:poles}; however, when the $v_{i}$ have mixed signs,
then for certain values of $\alpha$ ($\alpha>-4$ in Fig. \eqref{fig:Poles})
there are $m-2$ real roots and a complex conjugate pair. 

This follows because between any pair of consecutive $v_{i}$'s that
are both negative (positive), the function in Eq. \eqref{eq:poles} goes
from negative (positive) to positive (negative) infinity. Thus there
is at least $m-2$ real solutions to Eq. \eqref{eq:poles}. Therefore
there are at most a complex conjugate pair of solutions. When a complex
conjugate pair of solutions exist, they correspond to the solution
of Eq. \eqref{eq:poles} where $v_{i}$ changes sign (see Fig. \eqref{fig:Poles}). 

The non-existence of a complex conjugate pair means lack of support
in the distribution of the $N-$fold sum. In Plemelj-Sokhotsky formula
the imaginary part needs to be taken. Lastly note, that there are
at most one pair of complex conjugate roots to Eq. \eqref{eq:poles}.
In other words, the roots are either real (i.e. zero probability in
the density) or have at most a complex conjugate pair. 

How would one find the roots? There exists a matrix such that its
eigenvalues are the roots (set of $w$ that are the zeros of Eq. \eqref{eq:poles})
of the above
\begin{equation}
\text{diag}(v_{1},\dots,v_{m})+\frac{1}{\alpha}uv^{T}\equiv\text{diag}(v)+\frac{1}{\alpha}uv^{T},
\end{equation}
which is a general rank-one update, where $u=[1,\dots,1]^{T}$ is
a column vector of length $m$. This is the non-symmetric generalization
of the more standard secular equations method \cite{trefethen1997numerical}.

To see this, assume non-singularity, which yields $\det\left(\text{diag}(v)+\frac{uv^{T}}{\alpha}-\lambda I\right)=0$
or $\det\left(I+\left(\text{diag}(v)-\lambda I\right)^{-1}\frac{uv^{T}}{\alpha}\right)=0$;
therefore (using trace properties): $1+\frac{1}{\alpha}v^{T}\left(\text{diag}(v)-wI\right)^{-1}u=0$.
Writing it out we have:
\[
\frac{1}{\alpha}\sum_{i=1}^{m}\frac{v_{i}}{v_{i}-w}=-1
\]
Therefore, eigenvalues of $\text{diag}(v)+\frac{1}{\alpha}uv^{T}$
give the roots that we were seeking \footnote{We can just generate the Matlab code by: $\left(\text{diag}(v)+\frac{1}{\alpha}uv^{T}\right)y=\mathtt{v.*y+u*dot(v,y)/}\alpha$;
where $\mathtt{v}=[v_{1}\cdots v_{m}];$}. 
\begin{figure}
\includegraphics[scale=0.22]{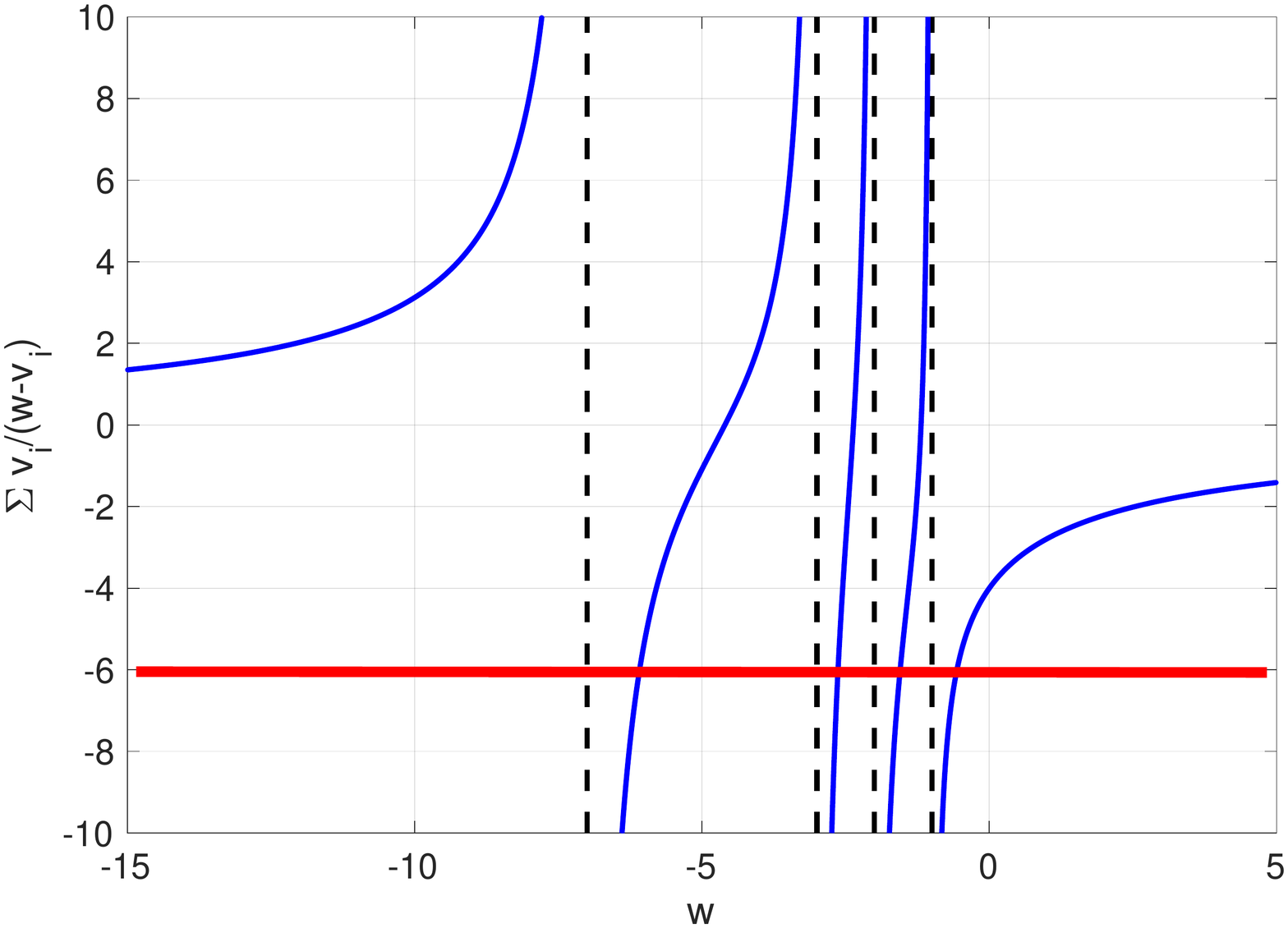}\includegraphics[scale=0.22]{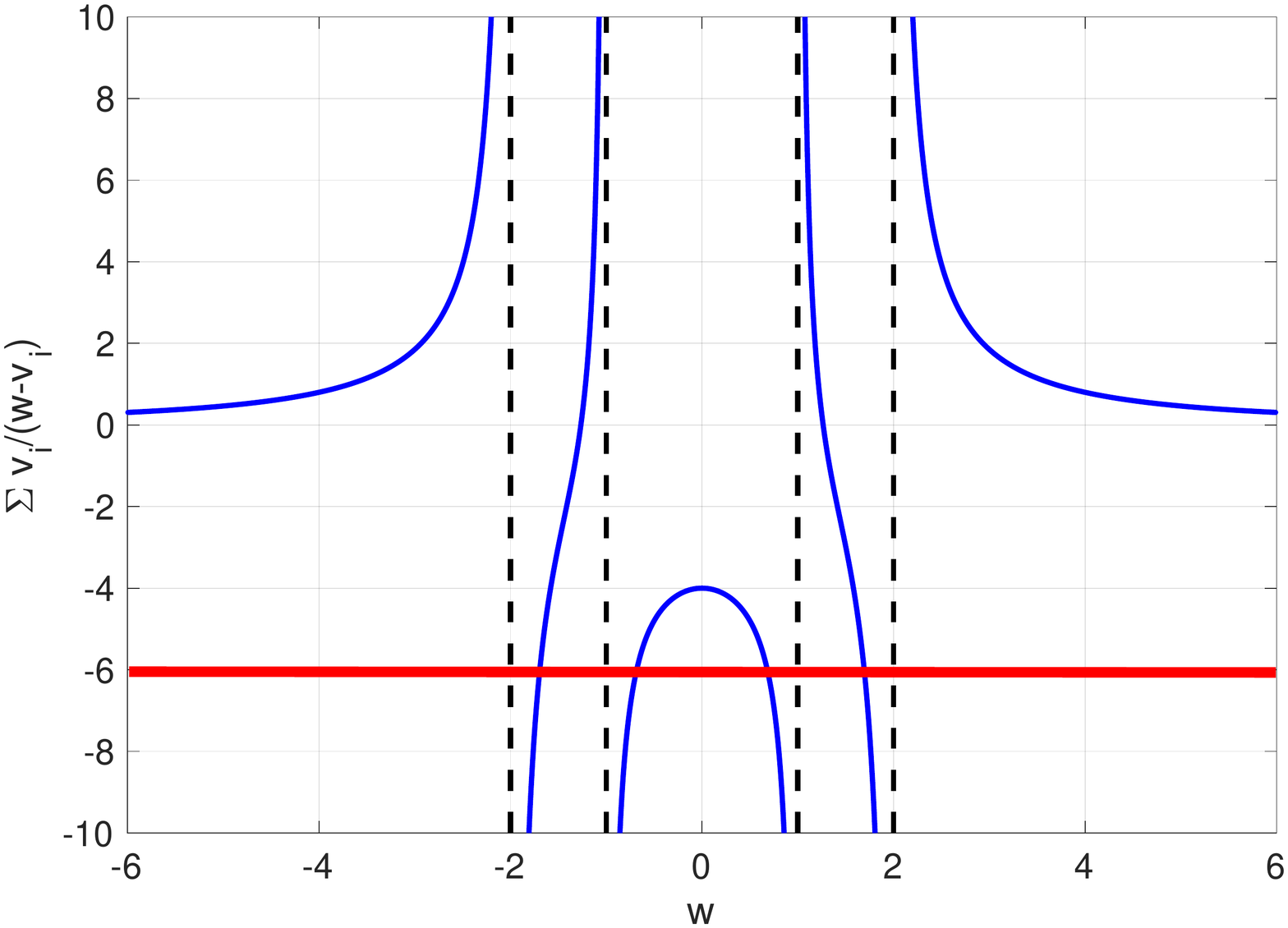}\includegraphics[scale=0.22]{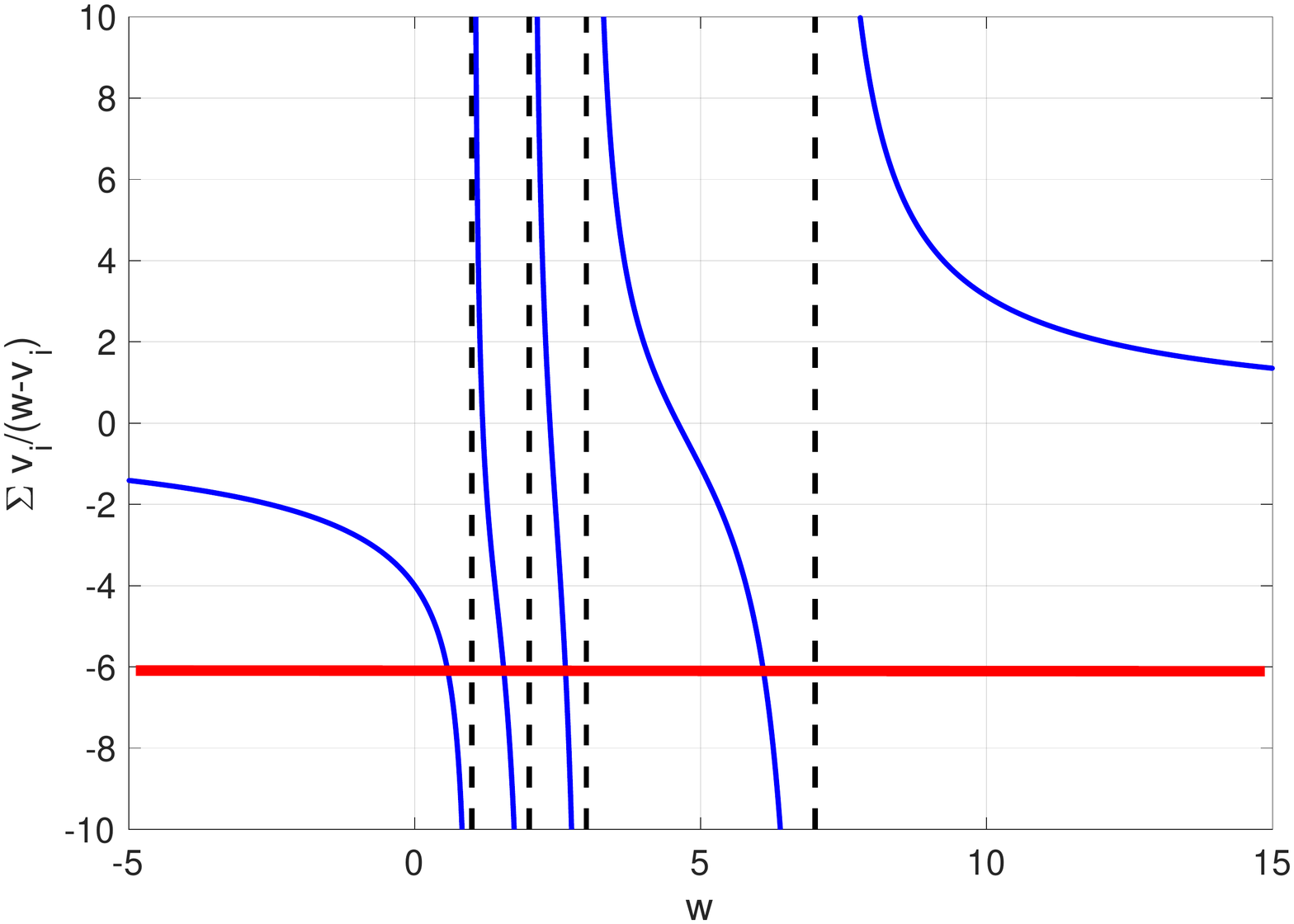}\caption{\label{fig:Poles}The red horizontal line is at $\alpha=-6$. Left:
$v=[-1,-2,-3,-7]$, Middle: $v=[-2,-1,1,2]$, has mixed signs, Right:
$v=[1,2,3,7]$}
\end{figure}
\begin{rem}
It seems possible that one can compute the complex eigenvalues efficiently
for an interval of different $z$ values by performing one initial
computation, obtain the $2-$ dimensional eigenspace for a complex
pair, and then update only that space with different $z$ values by
an Arnoldi method.
\end{rem}
\section{\label{sec:Illustrations-and-Applications}Illustrations and Applications}
For majority of applications involving non-commuting matrices, we
believe, free probability theory suffices. However, when the latter
fails, we have found that a combination of the two extreme approximations
(i.e., free and classical) to work very well. In particular, under
rather very mild conditions the natural parameter, $p$, for a convex
combination is obtained by matching fourth moments. Below we illustrate
the theory using some examples. 

Let us push the analytical calculation of $p$. Using Eq. \eqref{eq:pExplicit-1}
we have 
\[
p=\frac{\mathbb{E}\text{Tr}\left[M_{c}^{4}\right]-\mathbb{E}\text{Tr}\left[M^{4}\right]}{\mathbb{E}\text{Tr}\left[M_{c}^{4}\right]-\mathbb{E}\text{Tr}\left[M_{f}^{4}\right]},
\]
which by Eq. \eqref{eq:fourthMoment}, and noting that in the classical
approximation the summands commute, reads
\begin{equation}
p=\frac{\mathbb{E}\text{Tr}\left[M_{1}^{2}M_{2}^{2}\right]-\mathbb{E}\text{Tr}\left[\left(M_{1}M_{2}\right)^{2}\right]}{\mathbb{E}\text{Tr}\left[M_{1}^{2}M_{2}^{2}\right]-\mathbb{E}\text{Tr}\left[\left(\Lambda_{1}Q_{\beta}^{\dagger}\Lambda_{2}Q_{\beta}\right)^{2}\right]},\label{eq:p_practice}
\end{equation}
where with no loss of generality we take $M_{1}=\text{diag}(\lambda_{1},\dots,\lambda_{m})$
and we have
\[
\begin{array}{ccccc}
\text{Classical} & : & \mathbb{E}\text{Tr}\left[M_{1}^{2}M_{2}^{2}\right] & = & \mathbb{E}\sum_{i,j}\lambda_{i}^{2}|b_{ij}|^{2}\\
\text{Exact} & : & \mathbb{E}\text{Tr}\left[\left(M_{1}M_{2}\right)^{2}\right] & = & \mathbb{E}\sum_{i,j}\lambda_{i}b_{ij}\lambda_{j}b_{ji}=\mathbb{E}\sum_{i,j}\lambda_{i}\lambda_{j}|b_{ij}|^{2},\\
\text{Free} & : & \mathbb{E}\text{Tr}\left[\left(M_{1}Q_{\beta}^{\dagger}\Lambda_{2}Q_{\beta}\right)^{2}\right] & = & \mathbb{E}\sum_{i,j,k,p}\lambda_{i}\bar{q}_{ji}\mu_{j}q_{jk}\lambda_{k}\bar{q}_{pk}\mu_{p}q_{pi},
\end{array}
\]
where for the Free approximation of $M_{2}$ we substituted $Q_{\beta}^{\dagger}\Lambda_{2}Q_{\beta}$
and recall that $\Lambda_{2}=\text{diag}(\mu_{1},\dots,\mu_{m})$.
It is useful to further the computation of the Free approximation
\begin{equation}
\begin{array}{cccc}
\mathbb{E}\text{Tr}\left[\left(M_{1}Q_{\beta}^{\dagger}\Lambda_{2}Q_{\beta}\right)^{2}\right] & = & \mathbb{E}\sum_{i\ne k,j\ne p}\lambda_{i}\lambda_{k}\mu_{j}\mu_{p}\text{ }\bar{q}_{ji}q_{jk}\bar{q}_{pk}q_{pi} & "i\ne k","j\ne p"\\
 & + & \mathbb{E}\sum_{k,j\ne p}\lambda_{k}^{2}\text{ }\mu_{j}\mu_{p}\text{ }|q_{jk}|^{2}|q_{pk}|^{2} & \quad"i=k","j\ne p"\\
 & + & \mathbb{E}\sum_{i\ne k,j}\lambda_{i}\lambda_{k}\text{ }\mu_{j}^{2}\text{ }|q_{jk}|^{2}\text{ }|q_{ji}|^{2} & \quad"i\ne k","j=p"\\
 & + & \mathbb{E}\sum_{j,k}\lambda_{k}^{2}\text{ }\mu_{j}^{2}\text{ }|q_{jk}|^{4} & \quad"i=k","j=p".
\end{array}\label{eq:ETr_FreeDepart}
\end{equation}
\subsection{Sum of a diagonal and a block diagonal matrix}
Let $m=64$. As before and with no loss of generality we take $M_{1}$
to be diagonal. Let $M_{1}=\text{diag}(\lambda_{1},\dots,\lambda_{m})$
with $\lambda_{i}\sim{\cal N}(0,1)$, and let $M_{2}$ the block diagonal
matrix:
\begin{figure}
\begin{centering}
\includegraphics[scale=0.32]{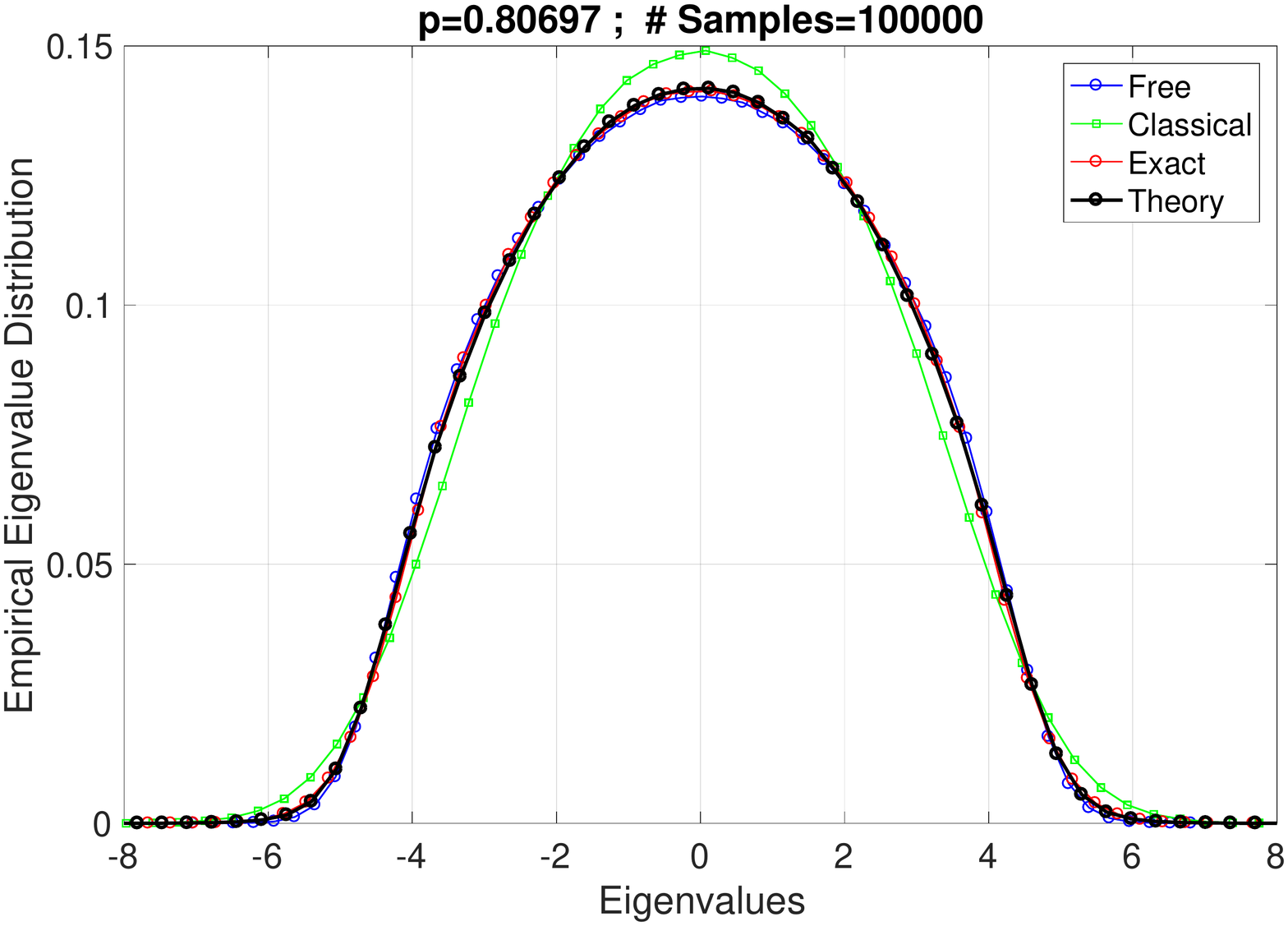}\includegraphics[scale=0.32]{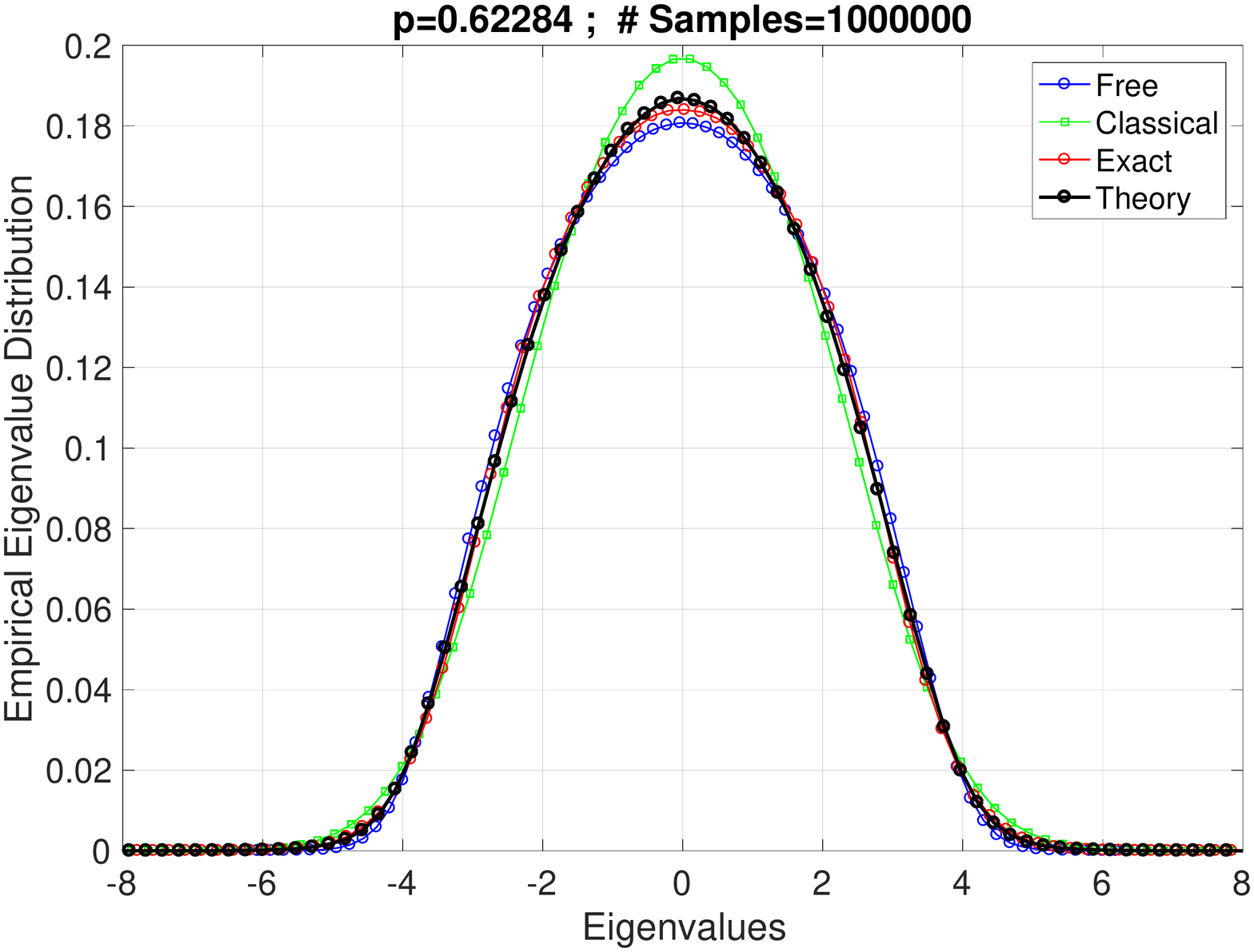}
\par\end{centering}
\caption{\label{fig:illustration}Sum of a diagonal random matrix with normally
distributed diagonal entries with a block diagonal matrix whose diagonal
blocks are independently drawn from the $\ell\times\ell$ GOE. Left:
$\ell=8$ and Right: $\ell=4$. In each sample we generate new matrices.
Compare the empirical $p$ values on the plots with the theoretical
that one obtains from Eq. \eqref{eq:p_Exact_Example}, which for $\ell=8$
is $p=0.8048$ and for $\ell=4$ is $p=0.62264$.}
\end{figure}
\begin{equation}
M_{2}=\left[\begin{array}{cccc}
B_{1}\\
 & B_{2}\\
 &  & \ddots\\
 &  &  & B_{k}
\end{array}\right]\label{eq:M2}
\end{equation}
where each $B_{k}$ is an $\ell\times\ell$ independent GOE matrix
with $k\ell=m$ (see Fig. \eqref{fig:illustration}). We illustrate the technique
with $a_{i}\sim{\cal N}(0,1)$. In Fig. \eqref{fig:illustration} we
plot the eigenvalue distribution based on samples of $M_{1}$ and
$M_{2}$ as indicated on the plots for $\ell=8$ and $\ell=4$. Numerically,
in each sample we obtain each $B_{i}$ by first generating an $\ell\times\ell$
random real gaussian matrix $G_{i}$, whose entries are standard normals
and then define $B_{i}$
\[
B_{i}=\frac{G_{i}^{T}+G_{i}}{2}.
\]

This is an example for which the relative eigenvectors have a block-diagonal
structure and therefore do not satisfy the uniformity property in
Assumption \eqref{assu:permIndep}. 

Below we derive formulas for general matrices of size $m$ with $k$
blocks of size $\ell\times\ell$ (clearly $k\ell=m$) and for general
$\beta$.

From the above, and using the fact that $M_{1}$ and $M_{2}$ are
independent, it is easy to see that $\mathbb{E}\text{Tr}[M_{1}^{2}M_{2}^{2}]=\sum_{i,j}\mathbb{E}(\lambda_{i}^{2})\mathbb{E}(|b_{ij}|^{2})$.
If $\lambda_{i}\sim{\cal N}(0,1)$, then $\mathbb{E}(\lambda_{i}^{2})=1$
for all $i$. Moreover, since the total number of nonzero diagonal
terms in $B$ is $m$ and the total number of nonzero diagonal terms
is $k\ell(\ell-1)=m(\ell-1)$ we have 
\begin{eqnarray*}
\mathbb{E}(|b_{ij}|^{2}) & = & \frac{1}{m^{2}}\sum_{i,j}|b_{ij}|^{2}=\frac{1}{m^{2}}\left\{ m\mathbb{E}(b_{11}^{2})+m(\ell-1)\mathbb{E}(|b_{12}|^{2})\right\} \\
 & = & \frac{1}{m}\left[1+\left(\ell-1\right)\frac{\beta}{2}\right]\qquad\text{For G(O/U/S)E block-diag. matrix}.
\end{eqnarray*}
because for the G(O/U/S)E matrix, the variance of any diagonal entry
is clearly $1$ and any off diagonal entry is $\beta/2$. Therefore
the classical answer is $\mathbb{E}\text{Tr}[M_{1}^{2}M_{2}^{2}]=\sum_{i,j}\mathbb{E}(|b_{ij}|^{2})=m\left(1+\beta(\ell-1)/2\right)$.

Let us now calculate, the exact departing term. By the independence
of $M_{1}$ and $M_{2}$ and since $\mathbb{E}(\lambda_{i}\lambda_{j})=\delta_{i,j}\mathbb{E}(\lambda_{i}^{2})$,
we have $\mathbb{E}\text{Tr}[(M_{1}M_{2})^{2}]=\mathbb{E}\sum_{i,j}\lambda_{i}\lambda_{j}|b_{ij}|^{2}=k\mathbb{E}\sum_{1\le i,j\le\ell}\lambda_{i}\lambda_{j}|b_{ij}|^{2}=k\left\{ \ell\mathbb{E}(\lambda_{i}^{2})\mathbb{E}(|b_{ii}|^{2})\delta_{i,j}\right\} =k\ell=m$. 

We now turn to the corresponding quantity in the free approximation.
In the formulas above (Eq. \eqref{eq:ETr_FreeDepart}) we need $\mathbb{E}(\mu_{i})$
, $\mathbb{E}(\mu_{i}\mu_{j})$ and $\mathbb{E}(\mu_{i}^{2})$, where
now $\mu_{i}$ denotes an eigenvalue. For the G(O/U/S)E, $\mathbb{E}(\lambda_{i})=0$.
Denoting by $||\centerdot||_{F}$ the Frobenius norm, for any $\ell\times\ell$
G(O/U/S)E matrix $B_{k}$ we have $\mathbb{E}(\mu_{i}^{2})=\frac{1}{\ell}\mathbb{E}||B_{k}||_{F}$,
but $||B_{k}||_{F}=\frac{1}{\ell}\mathbb{E}\left\{ \sum_{1\le i,j\le\ell}|b_{i,j}|^{2}\right\} =\frac{1}{\ell}\{\ell\mathbb{E}(|b_{i,i}|^{2}+\ell(\ell-1)\mathbb{E}(|b_{i,j}|^{2})\}$.
We conclude that $\mathbb{E}(\mu_{i}^{2})=\left(1+\beta(\ell-1)/2\right)$.
However, the size of $Q_{\beta}$ matrix is still $m$. To calculate
$\mathbb{E}(\mu_{i}\mu_{j})$ for $j\ne i$ note that 
\[
\mathbb{E}[\text{Tr}(B_{k})\text{Tr}(B_{k})]=\mathbb{E}\sum_{i=1}^{\ell}\mu_{i}^{2}+\mathbb{E}\sum_{i\ne j}\mu_{i}\mu_{j}
\]
But $X_{\ell}\equiv\text{Tr}(B_{k})$ is a sum of $\ell$ independent
standard normal variables, which has mean zero and variance $\ell$.
Moreover, by independence and zero mean, the cross terms are zero
and we have $\mathbb{E}[\text{Tr}(B_{k})\text{Tr}(B_{k})]=\mathbb{E}[X_{\ell}^{2}]=\ell$.
Lastly, we just derived $\mathbb{E}(\mu_{i}^{2})$, so we have
\begin{eqnarray*}
\mathbb{E}(\mu_{i}^{2}) & = & 1+\frac{\beta(\ell-1)}{2},\\
\mathbb{E}\sum_{i\ne j}\mu_{i}\mu_{j} & = & \ell(\ell-1)\beta/2\implies\mathbb{E}(\mu_{i}\mu_{j})=-\frac{\beta}{2}.\qquad\text{For G(O/U/S)E}
\end{eqnarray*}
Comment: For $\mathbb{E}(\mu_{i}\mu_{j})$ the size of the matrix
and its blocks are irrelevant. 

Because of independence of $M_{1}$ from $M_{2}$ and $\mathbb{E}(\lambda_{i}\lambda_{j})=\delta_{i,j}$, the first and third sums in Eq. \eqref{eq:ETr_FreeDepart} vanish.
Moreover by the independence of eigenvalues from eigenvectors the
expectation is taken term-wise as
\begin{eqnarray*}
\mathbb{E}\text{Tr}\left[\left(M_{1}Q_{\beta}^{\dagger}\Lambda_{2}Q_{\beta}\right)^{2}\right] & = & \frac{\beta}{(m\beta+2)}\sum_{j\ne p}\text{ }\mathbb{E}(\mu_{j}\mu_{p})+\frac{\beta+2}{(m\beta+2)}\sum_{j}\mathbb{E}(\mu_{j}^{2})\\
 & = & \frac{m(\ell\beta+2)}{m\beta+2},
\end{eqnarray*}
because $\sum_{j\ne p}\mathbb{E}(\mu_{j}\mu_{p})=k\sum_{1\le j\ne p\le\ell}\mathbb{E}(\mu_{j}\mu_{p})$
and Weingarten formulas (see Eq. \eqref{eq:IPR_free} and the Table
below it).

Comment: The analytically derived values for $\mathbb{E}(|b_{ij}|^{2})$,
$\mathbb{E}(\mu_{i}\mu_{j})$, $\mathbb{E}(\mu_{j}^{2})$, and $\mathbb{E}\text{Tr}\left[\left(M_{1}Q_{\beta}^{\dagger}\Lambda_{2}Q_{\beta}\right)^{2}\right]$
were all checked against numerics with high accuracy. 

We can now analytically obtain $p$ (Eq. \eqref{eq:p_practice}) for
this problem to be
\begin{equation}
p=\frac{(\ell-1)(m\beta+2)}{(\ell-1)(m\beta+2)+2(m-\ell)}.\label{eq:p_Exact_Example}
\end{equation}
\begin{rem}
Note that when $\ell=m$, $p=1$ as expected the free answer becomes
exact when one sums a random diagonal matrix with an $m\times m$
 G(O/U/S)E. Also note the remarkable agreement of theoretical (Eq.
\eqref{eq:p_Exact_Example}) and empirical $p$ values in Fig. \eqref{fig:illustration} 
\end{rem}
If we were to use Eq. \eqref{eq:pTheorem} we would obtain for the example
in this section $p=1-m\mathbb{E}[|q_{s}|^{4}]=(\ell-1)/(\ell+2)$.
The reason there is a discrepancy with Eq. \eqref{eq:p_Exact_Example}
is that the block diagonal matrix $M_{2}$ does not obey  Assumption
\eqref{assu:permIndep}. 
\subsection{Sum of a diagonal with fixed Kac-Mudrock-Szego or Laplacian matrix}
Next we take the diagonal entries of $M_{1}$ to be $a_{i}\in[-1,+1]$
and take $M_{2}=K$, where $Q$ is a Haar orthogonal matrix and $K$
is the Kac-Mudrock-Szego matrix, whose entries, denoted by $k_{i,j}$,
are
\begin{figure}
\centering{}\includegraphics[scale=0.3]{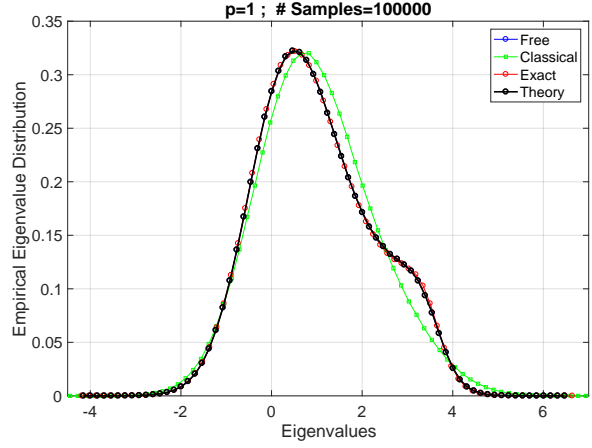}\caption{\label{fig:Kac-Mudrock-Szego}Sum of a diagonal random matrix, $M_{1}$,
with standard normal entries and the matrix, $M_{2}=Q^{T}KQ$, where
$K$ is taken to be Kac-Mudrock-Szego. Note that the eigenvalues of
$M_{2}$ are fixed.}
\end{figure}
\begin{figure}
\includegraphics[scale=0.3]{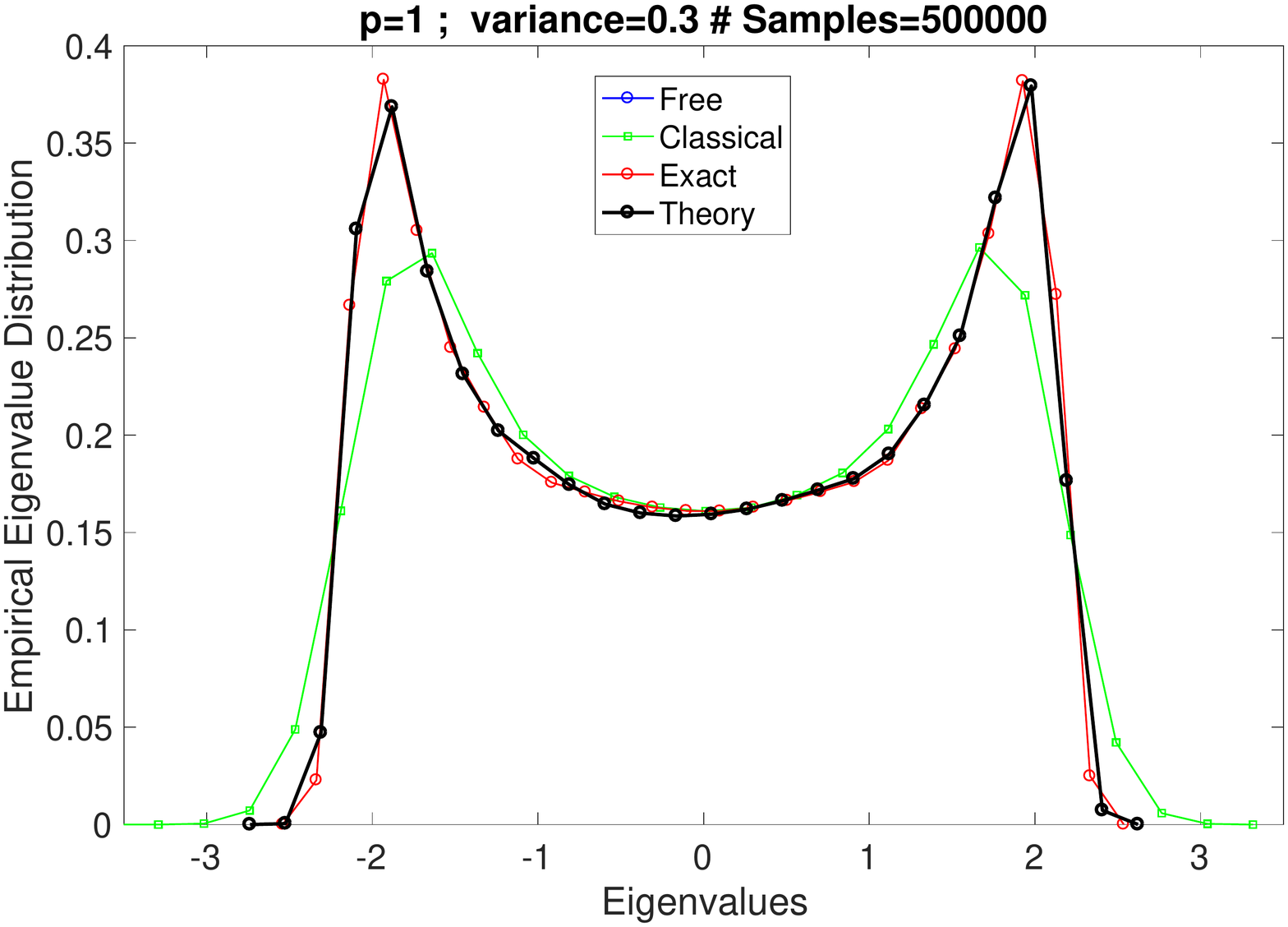}\includegraphics[scale=0.3]{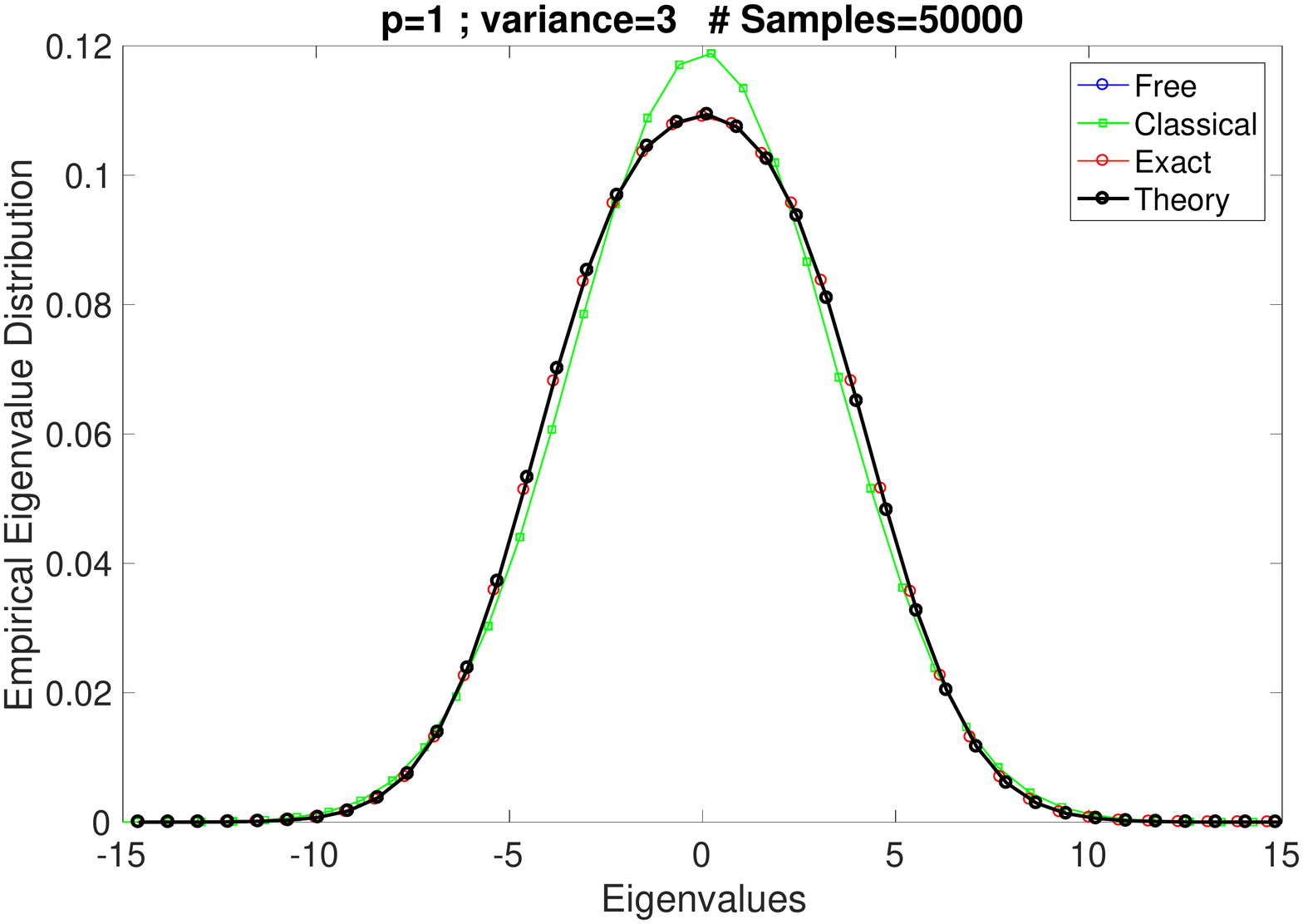}
\caption{\label{fig:illustration2}Left: (Anderson model) with variance $\text{var}(\lambda_{i})=0.3$.
Sum of a diagonal random gaussian matrix, $M_{1}$, and the hopping
matrix, $M_{2}=Q^{T}(2\mathbb{I}+L)Q$, where $L$ is the Laplacian
matrix. Right: (Anderson model) with $\text{var}(\lambda_{i})=3$.
Note that in all these plots examples $M_{2}$ has a fully deterministic
set of eigenvalues. In Fig. \eqref{fig:Illustration000} we show the
case where $\text{var}(\lambda_{i})=1$. }
\end{figure}
\[
k_{ij}=\rho^{|i-j|}
\]
where we take $\rho=1/2$; it can be shown that when $0<\rho<1$ then
$M_{2}\ge0$. We show the eigenvalues of the sum in Fig. \eqref{fig:Kac-Mudrock-Szego}. 

Lastly, we illustrate how well the density of states of the Anderson
model is captured by this technique. In this case $M_{1}=\text{diag}(a_{1},\dots,a_{n})$,
where $a_{i}$'s are independent standard gaussians and $M_{2}$ is
the nearest neighbors hopping matrix with periodic boundary conditions
\[
M_{2}=\left[\begin{array}{ccccc}
0 & 1 &  &  & 1\\
1 & 0 & 1\\
 & 1 & 0 & \ddots\\
 &  & \ddots & \ddots & 1\\
1 &  &  & 1 & 0
\end{array}\right].
\]

$M_{2}$ is equal to a shifted Laplacian matrix, where $M_{2}=2\mathbb{I}+L$,
where $\mathbb{I}$ is the identity and $L$ is the Laplacian matrix.
Elsewhere, we took $a_{i}$ to be randomly distributed from the semi-circle
law and proved that $M_{1}$ and $M_{2}$ have moments matching up
to $8$ \cite{chen2012error}. We showed that the method is successful
across the range of the strength of disorder (see Fig. \eqref{fig:Illustration000}
and Fig. \eqref{fig:illustration2}). Like in there we find that the
free approximation alone is quite adequate.

Comment: If one sets to find $p$ numerically by matching fourth moments,
one should note that the kurtosis can be very slow to converge. In
principle, if two matrices are free, one could numerically observe
a $p>1$ or if the classical end is the exact theory then $p<0$ can
be observed. These are byproducts of numerical inaccuracies of computing
the kurtoses.
\section{An Application: Density of State of Generic Local Quantum Spin Chains}
The density of states encodes useful information about the physics
of many-body systems. Here we apply our technique to quantum many-body
systems with generic interactions \cite{movassagh2011density,movassagh_IE2010}.
 Consider the Hamiltonian acting on the joint Hilbert space of $n$
$d-$dimensional complex vector spaces (e.g., spin $s$ particles,
where $d=2s+1$). The joint Hilbert space is $\left(\mathbb{C}^{d}\right)^{\otimes n}$
and the nearest neighbor interactions is given by the Hamiltonian
\begin{equation}
H=\sum_{k=1}^{n-1}\mathbb{I}_{d^{k-1}}\otimes H_{k,k+1}\otimes\mathbb{I}_{d^{n-k-1}}\label{eq:Hamiltonian_SpinChain}
\end{equation}
where each $H_{k,k+1}$ is a $d^{2}\times d^{2}$ matrix that we take
to be generic. For example, the local interactions can be distributed
according to GUE, or be random projectors, or Wishart matrices etc. 

The problem statement is then: Suppose the eigenvalue distribution
of $H_{k,k+1}$ is known, what is the eigenvalue distribution of $H$?

The exact problem is NP-Complete \cite{brown2011computational}. There
are two main sources of difficulties: 1. The size of the matrix $H$
is $d^{n}\times d^{n}$, which makes the exact diagonalization difficult
even for moderate sized problems. 2. Any two consecutive terms in
Eq. \eqref{eq:Hamiltonian_SpinChain} do not commute. 

Despite these challenges and the NP-completeness of the exact result,
the method described above provides an excellent approximation to
the true distribution. We now proceed to detail the results corroborated
with various numerical illustrations. 

In Eq. \eqref{eq:Hamiltonian_SpinChain} the summands with $k$ odd
all commute. Similarly the summands with $k$ even all commute. This
enables us to write $H$ in Eq. \eqref{eq:Hamiltonian_SpinChain} as
\begin{equation}
H_{odd}+H_{even},\label{eq:problem_Quantum}
\end{equation}
where each $H_{odd}$ and $H_{even}$ is $d^{n}$ dimensional and
is given by
\begin{eqnarray*}
H_{odd} & = & \sum_{k\text{ odd}}\mathbb{I}_{d^{k-1}}\otimes H_{k,k+1}\otimes\mathbb{I}_{d^{n-k-1}},\\
H_{even} & = & \sum_{k\text{ even}}\mathbb{I}_{d^{k-1}}\otimes H_{k,k+1}\otimes\mathbb{I}_{d^{n-k-1}}.
\end{eqnarray*}

We take $H_{odd}$ and $H_{even}$ as our two known matrices, where
an eigenvalue decomposition gives 
\begin{align*}
H_{odd} & =U_{odd}^{\dagger}\Lambda_{odd}U_{odd}\\
H_{even} & =U_{even}^{\dagger}\Lambda_{even}U_{even}.
\end{align*}
The unitary matrices of eigenvectors, $U_{odd}$ and $U_{even}$,
are (for an odd sized chain)
\begin{eqnarray*}
U_{odd} & = & U_{1,2}\otimes U_{3,4}\otimes\cdots\otimes U_{n-2,n-1}\otimes\mathbb{I}_{d}\\
U_{even} & = & \mathbb{I}_{d}\otimes U_{2,3}\otimes U_{4,5}\otimes\cdots\otimes U_{n-1,n}.
\end{eqnarray*}
In these equations $U_{k,k+1}$ denotes the eigenvector matrix of
$H_{k,k+1}$ and is therefore $d^{2}\times d^{2}$ in size.

The diagonal real matrices of eigenvalues $\Lambda_{odd}$ and $\Lambda_{even}$
are
\begin{figure}
\begin{raggedright}
\includegraphics[scale=0.28]{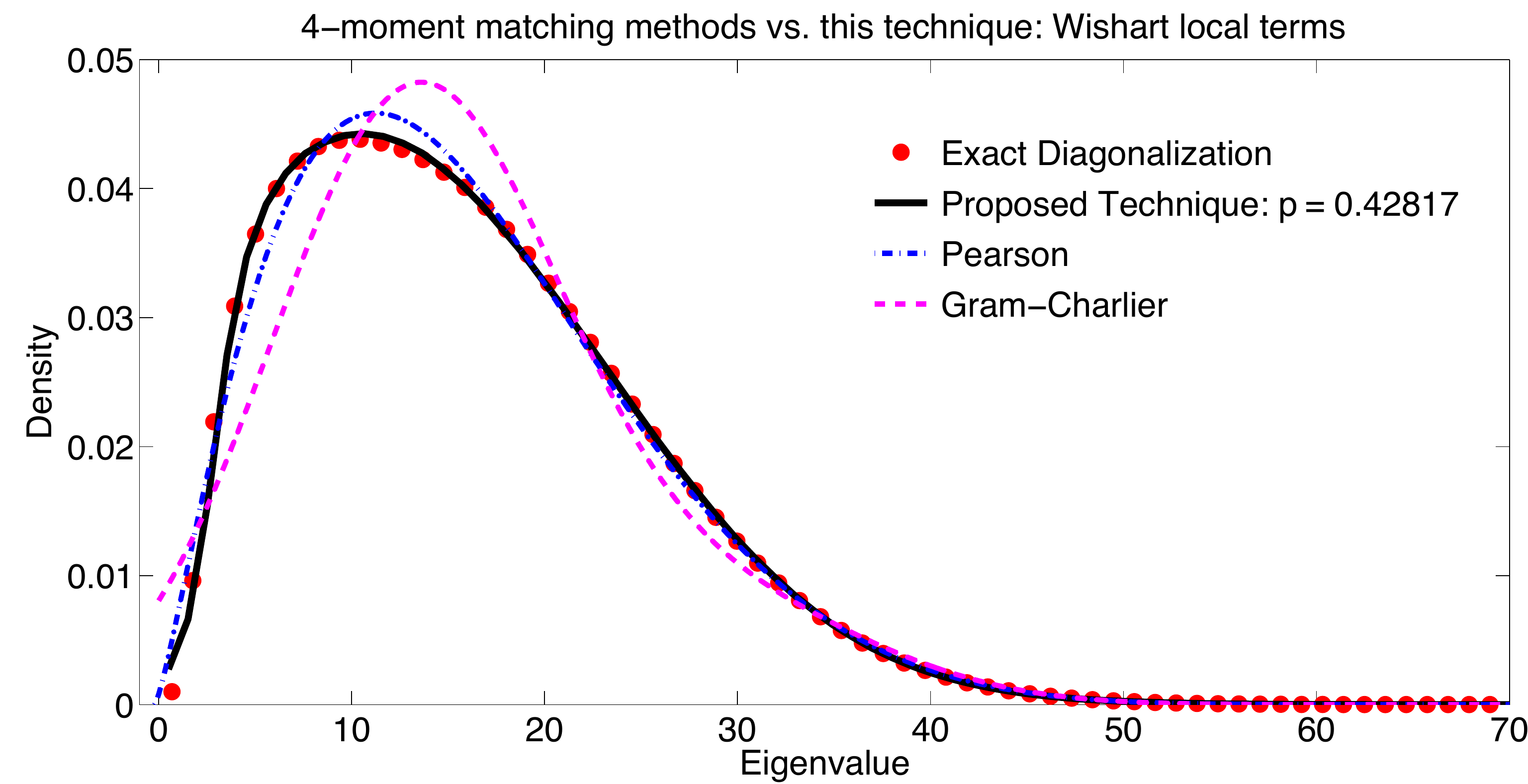}\includegraphics[scale=0.28]{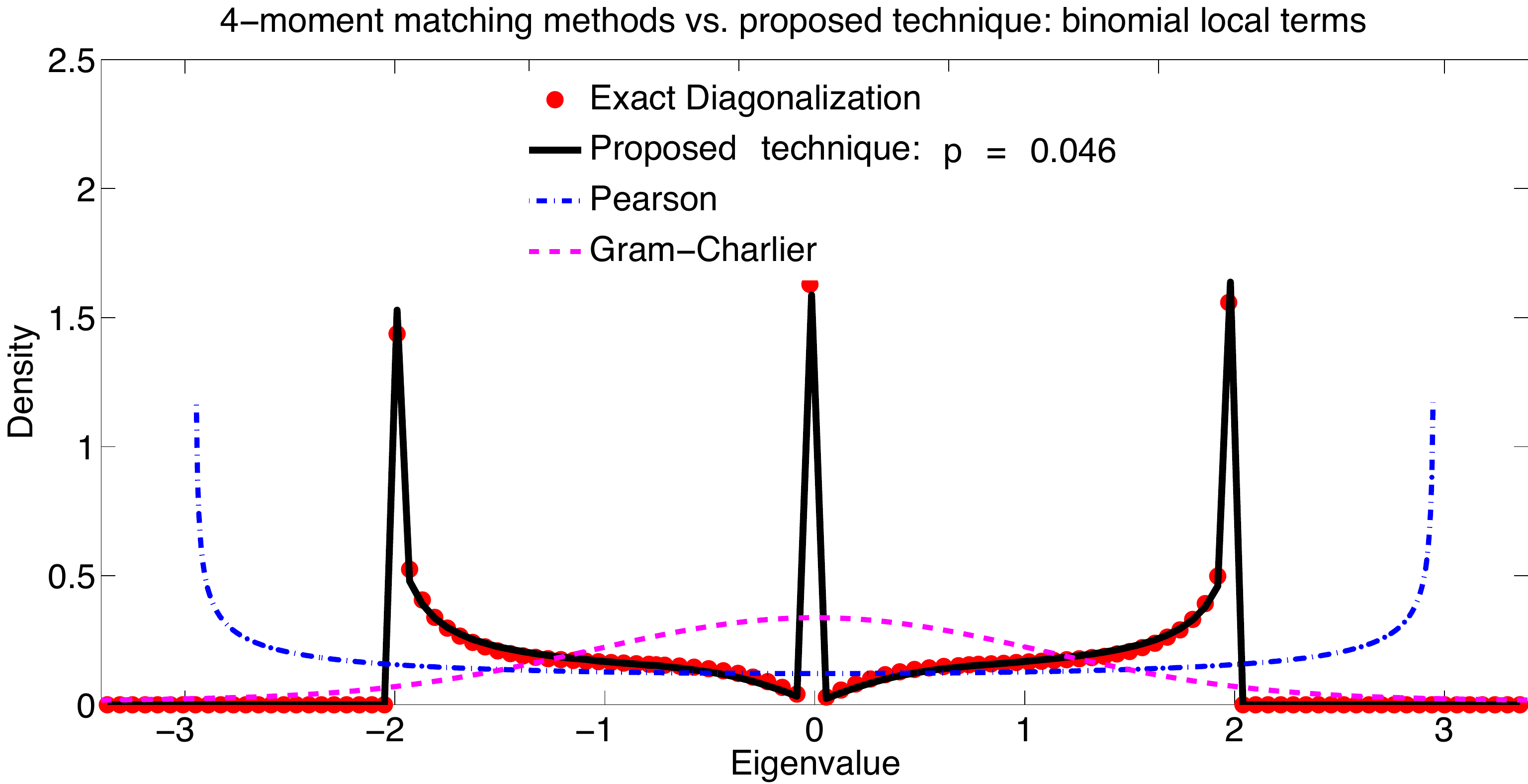}
\par\end{raggedright}
\caption{\label{fig:Comparison-of-GramCharlier}Comparison of our technique
with Gram-Charlier and Beta-ensemble. The left corresponds to what
was shown in Fig. \eqref{fig:Illustration000} and the right is a spin
chain of length $3$ where each of the two local interactions has
Haar eigenvectors and Bernoulli $\pm1$ eigenvalues. In the latter,
we took the size of each local term to be $25\times25$; therefore
$H$ is $125\times125$. }
\end{figure}
\begin{align*}
\Lambda_{odd} & =\sum_{k\text{ odd}}\mathbb{I}_{d^{k-1}}\otimes\Lambda_{k,k+1}\otimes\mathbb{I}_{d^{n-k-1}},\\
\Lambda_{even} & =\sum_{k\text{ even}}\mathbb{I}_{d^{k-1}}\otimes\Lambda_{k,k+1}\otimes\mathbb{I}_{d^{n-k-1}},
\end{align*}
and $\Lambda_{k,k+1}$ is the real and diagonal matrix of the eigenvalues
of $H_{k,k+1}$. The eigenvalues of $\Lambda_{odd}$ corresponds to
all possible sums of the eigenvalues of $\Lambda_{k,k+1}$, which
is easy to obtain. Similarly $\Lambda_{even}$ is easy to compute.

With no loss of generality we change basis in which $H_{odd}$ is
diagonal, whereby we have
\[
H=\Lambda_{odd}+Q_{s}^{-1}\Lambda_{even}Q_{s}
\]
and $Q_{s}\equiv U_{odd}U_{even}^{\dagger}$. The problem then is
to find a good approximation for the density of states of $H$. Recall
that we have two extreme ends that correspond to the classical and
free approximations
\begin{align*}
H_{c} & =\Lambda_{odd}+\Pi^{-1}\Lambda_{even}\Pi\\
H_{f} & =\Lambda_{odd}+Q^{-1}\Lambda_{even}Q
\end{align*}
where $\Pi$ and $Q$ are permutation and $\beta-$orthogonal Haar
matrices respectively (exactly as before). In the left figure of Fig. \eqref{fig:Illustration000} we showed the DOS for $H$, $H_{c}$ ,
$H_{f}$ and the proposed technique; the convex combination parameter
here is $p=0.43$. 

How does our technique compare to other known techniques? To the best
of our knowledge there are two note-worthy techniques that we can
compare against. The first is the Gram-Charlier expansion \cite{cramer2016mathematical}
which builds the distribution from the knowledge of first $k$ moments\footnote{A limitation of Gram-Charlier is that it can at times output a negative
densities.}. The second is a fit to the beta-distribution, which is part of MatLab's
library of function (pearson.m). Our technique, unlike the others,
seems to work much better than what one would expect from the knowledge
of the first four moments alone.
\section{Discussions: Limitations and Comparison}
In this paper we described a technique for calculating the eigenvalue
distribution of sums of matrices from the knowledge of the distribution
of the summands. The input to the theory is the known distribution
of the summands and the output is an approximation to the density
of the sum. We have laid out a step by step technology by which such
calculations can be carried out and provided an eigenvalue finding
subroutine which circumvents solving high order polynomials to solve
for the complex roots needed. We then compared our theory against
exact diagonalization. Through our numerical work we find that the
theory proposed gives excellent approximation of the exact eigenvalue
distributions in most cases.

The technique described above outputs an eigenvalue distribution,
which is a continuous curve or union of continuous curves. It is limited
in that it does not provide level spacing statistics (for $M$ in
Eq. \eqref{eq:problem-1}). For many problems of interest in physics,
such as quantum many-body systems, the difference between the smallest
two eigenvalues is of utmost importance. This difference is simply
called \textit{the gap}. Elsewhere we have proved that there is a
continuum of eigenvalues above the smallest eigenvalue \cite{movassagh2016generic}.
Although this implies that the gap tends to zero as $m\rightarrow\infty$
for generic (local) interactions and that we can quantify how it goes
to zero for gaussian ensembles, we do not have a detailed enough description
of eigenvalue spacings beyond.

Density of states does not necessarily provide information about 2-point
or higher order correlation functions. It would be interesting if
they were investigated.

We are aware of two other works (\cite{benaych2011continuous} and
\cite{bozejko1997q}) that formulate some form of interpolation between
a ``free\textquotedblright{} object and a ``classical\textquotedblright{}
object: In \cite{benaych2011continuous}, a random unitary matrix
is explicitly constructed through a Brownian motion process starting
at time $t=0$, and going to time $t=\infty$. ``Classical'' corresponds
to $t=0$, and ``free'' corresponds to $t=\infty$. The random unitary
matrix starts non-random and is randomized continuously until it fully
reaches Haar measure. In \cite{bozejko1997q}, through detailed combinatorial
constructions and investigation into Fock space representations of
Fermions and Bosons, unique measures are constructed that interpolate
between the limit of the classical central limit theorem, the gaussian,
and the free central limit theorem, the semicircle. The curve also
continues on to $t=-1$, which corresponds to two non-random atoms.

An unknown question is whether the unitary construction in \cite{benaych2011continuous}
leads to the same convolution interpolate as this paper where we take
a convex combination. Another unknown question is whether our proposal
and \cite{benaych2011continuous} lead to an analog of a limit of
a central limit theorem which would match that of \cite{bozejko1997q}.

We outline in the table below features found in each paper. The empty
boxes are opportunities for further research.\\
\begin{center}
\begin{tabular}{|c|c|c|c|c|}
\hline 
 & Application  & Unitary Matrix Construction  & Interpolate Convolution  & Iterate Convolution to a CLT\tabularnewline
\hline 
\hline 
This work & \CheckmarkBold{}  &  & \CheckmarkBold{}  & \tabularnewline
\hline 
\cite{benaych2011continuous} &  & \CheckmarkBold{}  & \CheckmarkBold{}  & \tabularnewline
\hline 
\cite{bozejko1997q} &  &  &  & \CheckmarkBold{} \tabularnewline
\hline 
\end{tabular}\\
\par\end{center}

\begin{flushleft}
Lastly, this work proposes a technique to obtain the eigenvalue distribution.
To ultimately understand the powers and limitations of
it, it would be most useful to take an applied perspective and apply
it to concrete problems.
\par\end{flushleft}

\section{Acknowledgements}
Some of this work was completed while RM had the support of the Simons
Foundation and the American Mathematical Society through the AMS-Simons
travel grant,  and IBM Research's support and freedom offered by his former
Herman Goldstine Fellowship. AE was supported by the National Science
Foundation through the grant  DMS-1312831.
\bibliographystyle{plain}
\bibliography{mybib}

\begin{thebibliography}{10}

\bibitem{benaych2011continuous}
Florent Benaych-Georges and Thierry L{\'e}vy.
\newblock A continuous semigroup of notions of independence between the
  classical and the free one.
\newblock {\em The Annals of Probability}, pages 904--938, 2011.

\bibitem{bozejko1997q}
Marek Bo{\.z}ejko, Burkhard K{\"u}mmerer, and Roland Speicher.
\newblock q-gaussian processes: non-commutative and classical aspects.
\newblock {\em Communications in Mathematical Physics}, 185(1):129--154, 1997.

\bibitem{brown2011computational}
Brielin Brown, Steven~T Flammia, and Norbert Schuch.
\newblock Computational difficulty of computing the density of states.
\newblock {\em Physical review letters}, 107(4):040501, 2011.

\bibitem{chen2012error}
Jiahao Chen, Eric Hontz, Jeremy Moix, Matthew Welborn, Troy Van~Voorhis,
  Alberto Su{\'a}rez, Ramis Movassagh, Alan Edelman, et~al.
\newblock Error analysis of free probability approximations to the density of
  states of disordered systems.
\newblock {\em Physical review letters}, 109(3):036403, 2012.

\bibitem{cramer2016mathematical}
Harald Cram{\'e}r.
\newblock {\em Mathematical Methods of Statistics (PMS-9)}, volume~9.
\newblock Princeton university press, 2016.

\bibitem{horn1962eigenvalues}
Alfred Horn.
\newblock Eigenvalues of sums of hermitian matrices.
\newblock {\em Pacific Journal of Mathematics}, 12(1):225--241, 1962.

\bibitem{klyachko1998stable}
Alexander~A Klyachko.
\newblock Stable bundles, representation theory and hermitian operators.
\newblock {\em Selecta Mathematica, New Series}, 4(3):419--445, 1998.

\bibitem{knutson2001honeycombs}
Allen Knutson and Terence Tao.
\newblock Honeycombs and sums of hermitian matrices.
\newblock {\em Notices Amer. Math. Soc}, 48(2), 2001.

\bibitem{movassagh_IE2010}
R~Movassagh and A~Edelman.
\newblock Isotropic entanglement.(2010).
\newblock {\em arXiv preprint arXiv:1012.5039}.

\bibitem{movassagh2016generic}
Ramis Movassagh.
\newblock Generic local hamiltonians are gapless.
\newblock {\em arXiv preprint arXiv:1606.09313}, 2016.

\bibitem{movassagh2011density}
Ramis Movassagh and Alan Edelman.
\newblock Density of states of quantum spin systems from isotropic
  entanglement.
\newblock {\em Physical review letters}, 107(9):097205, 2011.

\bibitem{nica2006lectures}
Alexandru Nica and Roland Speicher.
\newblock {\em Lectures on the combinatorics of free probability}, volume~13.
\newblock Cambridge University Press, 2006.

\bibitem{olver2012numerical}
Sheehan Olver and Raj~Rao Nadakuditi.
\newblock Numerical computation of convolutions in free probability theory.
\newblock {\em arXiv preprint arXiv:1203.1958}, 2012.

\bibitem{riordan2012introduction}
John Riordan.
\newblock {\em Introduction to combinatorial analysis}.
\newblock Courier Corporation, 2012.

\bibitem{trefethen1997numerical}
Lloyd~N Trefethen and David Bau~III.
\newblock {\em Numerical linear algebra}, volume~50.
\newblock Siam, 1997.

\bibitem{voiculescu1992free}
Dan~V Voiculescu, Ken~J Dykema, and Alexandru Nica.
\newblock {\em Free random variables}.
\newblock Number~1. American Mathematical Soc., 1992.

\bibitem{weyl1912asymptotische}
Hermann Weyl.
\newblock Das asymptotische verteilungsgesetz der eigenwerte linearer
  partieller differentialgleichungen (mit einer anwendung auf die theorie der
  hohlraumstrahlung).
\newblock {\em Mathematische Annalen}, 71(4):441--479, 1912.

\end{thebibliography}
\end{document}